%% file: Covertime.tex
\documentclass[10pt]{article}

\usepackage[margin=1.5in]{geometry}

\usepackage{bm}
\usepackage{amsmath, amssymb, amsthm}
\usepackage[labelfont=bf,font=small,indention=0.3cm,margin=0cm]{caption} 

\renewcommand{\labelitemi}{\textbullet}

\usepackage{extarrows}
\usepackage[numbers,sort&compress,longnamesfirst,sectionbib]{natbib}             
\shortcites{GhoshLMMPRRTZ99,AKLLR79,CRRST97,Priezzhev1996,Bampas09,YanovskiWB03}

\usepackage{tikz}

\usepackage{xspace}
\usepackage[latin1]{inputenc}
\usepackage{graphicx}
\usepackage{epsfig}
\usepackage{multirow}
\usepackage{wrapfig}
\usepackage{subfigure}      

\newcommand{\dist}{\operatorname{dist}}
\newcommand{\diam}{\operatorname{diam}}

\def\mod{\operatorname{mod}}

\renewcommand{\O}{\mathcal{O}}

\newtheorem{thm}{Theorem}  
\newtheorem{lem}[thm]{Lemma}

\newtheorem{cor}[thm]{Corollary}

\newtheorem{pro}[thm]{Proposition}

\newtheorem{defi}[thm]{Definition}

\numberwithin{thm}{section}

\newcommand{\thmref}[1]{Theorem~\ref{thm:#1}}

\newcommand{\thmrefss}[3]{Theorems~\ref{thm:#1},~\ref{thm:#2} and~\ref{thm:#3}}
\newcommand{\shortthmref}[1]{Thm.~\ref{thm:#1}}
\newcommand{\shortthmrefs}[2]{Thm.~\ref{thm:#1} and~\ref{thm:#2}}

\newcommand{\lemref}[1]{Lemma~\ref{lem:#1}}

\newcommand{\shortcorref}[1]{Cor.~\ref{cor:#1}}
\newcommand{\mycorref}[1]{Corollary~\ref{cor:#1}}

\newcommand{\defref}[1]{Definition~\ref{def:#1}}
\newcommand{\defrefs}[2]{Definitions~\ref{def:#1} and~\ref{def:#2}}

\newcommand{\figref}[1]{Figure~\ref{fig:#1}}

\newcommand{\tabref}[1]{Table~\ref{tab:#1}}
\newcommand{\tabrefs}[2]{Tables~\ref{tab:#1} and~\ref{tab:#2}}
\newcommand{\secref}[1]{Section~\ref{sec:#1}}

\newcommand{\eq}[1]{equation~\eqref{eq:#1}}
\newcommand{\eqs}[2]{equations~\eqref{eq:#1} and~\eqref{eq:#2}}

\newcommand{\mylem}[2]{\begin{lem}\label{lem:#1}#2\end{lem}}

\newcommand{\mythm}[2]{\begin{thm}\label{thm:#1}#2\end{thm}}

\newcommand{\Ex}[1]{\ensuremath{\operatorname{\mathbf{E}}\left[#1\right]}}
\newcommand{\Exbig}[1]{\ensuremath{\operatorname{\mathbf{E}}\big[#1\big]}}

\newcommand{\cN}{\ensuremath{\mathcal{N}}}
\newcommand{\cM}{\ensuremath{\mathcal{M}}}
\newcommand{\calC}{\ensuremath{\mathcal{C}}}
\newcommand{\cG}{\ensuremath{\mathcal{G}}}

\newcommand{\bP}{\ensuremath{\mathbf{P}}}
\newcommand{\bD}{\ensuremath{\mathbf{D}}}

\newcommand{\bZ}{\ensuremath{\mathbf{Z}}}

\newcommand{\bmdEC}{\ensuremath{\widetilde{\bm{\mathsf{EC}}}}}
\newcommand{\bmdVC}{\ensuremath{\widetilde{\bm{\mathsf{VC}}}}}
\newcommand{\bmrEC}{\ensuremath{\bm{\mathsf{EC}}}}
\newcommand{\bmrVC}{\ensuremath{\bm{\mathsf{VC}}}}

\newcommand{\dH}{\ensuremath{\widetilde{\mathsf{H}}}}
\newcommand{\dEC}{\ensuremath{\widetilde{\mathsf{EC}}}}
\newcommand{\dVC}{\ensuremath{\widetilde{\mathsf{VC}}}}  

\newcommand{\rH}{\ensuremath{\mathsf{H}}}
\newcommand{\rEC}{\ensuremath{\mathsf{EC}}}
\newcommand{\rVC}{\ensuremath{\mathsf{VC}}}

\newcommand{\tN}{\ensuremath{\widetilde{N}}}
\newcommand{\tx}{\ensuremath{\widetilde{x}}}
\newcommand{\ts}{\ensuremath{\widetilde{s}}}
\newcommand{\td}{\ensuremath{\widetilde{d}}}
\newcommand{\tr}{\ensuremath{\widetilde{r}}}
\newcommand{\tkappa}{\ensuremath{\widetilde{\kappa}}}

\newcommand{\kary}{$k$\nobreakdash-ary\xspace}

\newcommand{\Oh}{\mathcal{O}}

\newcommand{\tsum}{\textstyle\sum}

\DeclareSymbolFont{AMSb}{U}{msb}{m}{n}
\newcommand{\N}{{\mathbb{N}}}

\newcommand{\Z}{{\mathbb{Z}}}

\let\oldsqrt\sqrt
\def\hksqrt{\mathpalette\DHLhksqrt}
\def\DHLhksqrt#1#2{\setbox0=\hbox{$#1\oldsqrt{#2\,}$}\dimen0=\ht0
   \advance\dimen0-0.2\ht0
   \setbox2=\hbox{\vrule height\ht0 depth -\dimen0}%
   {\box0\lower0.4pt\box2}}
\renewcommand\sqrt\hksqrt
\renewcommand{\leq}{\leqslant}
\renewcommand{\geq}{\geqslant}

\renewcommand{\epsilon}{\varepsilon}

\allowdisplaybreaks[1]   
\title{The Cover Time of Deterministic Random Walks}

\author{Tobias Friedrich\\
\small Max-Planck-Institut für Informatik\\[-0.8ex]
\small Campus E1.4, 66123 Saarbrücken\\[-0.8ex]
\small Germany
\and
Thomas Sauerwald\\
\small Simon Fraser University\\[-0.8ex]
\small Burnaby B.C. V5A 1S6\\[-0.8ex]
\small Canada
}

\date{}

\begin{document}

\maketitle

\begin{abstract}
    \noindent
    The rotor router model is a popular deterministic analogue of a random walk on a
    graph. Instead of moving to a random neighbor, the neighbors are served
    in a fixed order.
    We examine how fast this ``deterministic random walk''
    covers all vertices (or all edges).
    We present general techniques to derive upper bounds for the vertex and edge
    cover time and derive matching lower bounds
    for several important graph classes.
    Depending on the topology, the deterministic random walk can be asymptotically
    faster, slower or equally fast as the classic random walk.
    We also examine the short term behavior of deterministic random walks,
    that is, the time to visit a fixed small number of vertices or edges.
\end{abstract}

\section{Introduction}

\noindent We examine the cover time of a simple deterministic process
known under various names such as ``rotor router model'' or ``Propp machine.''
It can be viewed as an attempt to derandomize random walks on graphs $G=(V,E)$.
In the model each vertex~$x\in V$ is equipped with
a ``rotor'' together with
a fixed sequence of the neighbors of~$x$ called ``rotor sequence.''
While a particle (chip, coin, \ldots) performing a random walk leaves
a vertex in a random direction,
the deterministic random walk always goes in the direction the rotor is pointing.
After a particle is sent, the rotor is updated to the next position of its rotor sequence.
We examine how fast this model covers all vertices and/or edges,
when one particle starts a walk from an arbitrary vertex.

\subsection{Deterministic random walks}

\noindent
The idea of rotor routing appeared independently several times in the literature.
First under the name ``Eulerian walker'' by \citet{Priezzhev1996},
then by \citet{WagnerLB99} as ``edge ant walk'' and
later by \citet{DumitriuTW03} as ``whirling tour.''
Around the same time it was also popularized
by James Propp~\citep{Kleber} and analyzed by \citet{CooperSpencer} who called it the ``Propp machine.''
Later the term ``deterministic random walk'' was established in Doerr et al.~\citep{1DPropp,2DPropp}.
For brevity, we omit the ``random'' and just refer to ``deterministic walk.''

\citet{CooperSpencer} showed the following remarkable
similarity between the expectation of a random walk and
a deterministic walk with cyclic rotor sequences:
If an (almost) arbitrary distribution of particles is
placed on the vertices of an infinite grid~$\Z^d$ and does a simultaneous walk in the
deterministic walk model, then at all times and on each vertex, the number of particles
deviates from the expected number the standard random walk would have gotten there,
by at most a constant.  This constant is precisely known for the cases
$d=1$~\citep{1DPropp} and~$d=2$~\citep{2DPropp}.
It is further known
that there is no such constant for infinite trees~\citep{TreePropp}.
\citet{LevinePeres09} also extensively studied a related model called
internal diffusion-limited aggregation~\citep{IDLA1,IDLA2}
for deterministic walks.

As in these works, our aim is to understand random walk and their deterministic counterpart from a theoretical viewpoint. However, we would like to mention that the rotor router
mechanism also led to improvements in applications.  With a random initial
rotor direction, the quasirandom rumor spreading protocol broadcasts
faster in some networks than its
random counterpart~\citep{DFS08,DFS09,DFKS09,ADHP2009}.
A similar idea is used in
quasirandom external mergesort~\citep{BarveGV97} and quasirandom load balancing~\citep{FGS10}.

We consider our model of a deterministic walk based on rotor routing to be a simple and canonic derandomization
of a random walk which is not tailored for search problems.
On the other hand, there is a vast literature on local deterministic agents/robots/ants patrolling or covering
all vertices or edges of a graph (e.g.~\citep{Rei08,GasieniecPRZ07,WagnerLB99,WagnerLB96,Korf90}).
For instance, \citet{CIKK09} studied a model where the walk uses adjacent edges which have been traversed the smallest number of times. However, all of these models are more specialized and require additional counters/identifiers/markers/pebbles on the vertices or edges of the explored graph.

\renewcommand{\arraystretch}{1.08}
\begin{table*}[t]
    \footnotesize
    \begin{center}
    \begin{tabular}{|l|cl|cl|}
    \hline
    \multirow{2}{*}{\bf Graph class $\bm G$} &
    \multicolumn{2}{|c|}{\bf Vertex cover time $\bmrVC\bm{(G)}$} &
    \multicolumn{2}{|c|}{
        \begin{minipage}[b]{0.001mm}\vspace*{.41cm}\end{minipage}
        \bf Vertex cover time $\bmdVC\bm{(G)}$
        \begin{minipage}[b]{0.001mm}\vspace*{.41cm}\end{minipage}
    }
    \\
    &
    \multicolumn{2}{|c|}{\bf of the random walk} &
    \multicolumn{2}{|c|}{\bf of the deterministic walk} \\
    \hline
    \hline
    \kary tree, $k=\Oh(1)$
    & $\Theta(n \log^2 n)$ & \cite[Cor.~9]{Zu92}
    & $\Theta(n \log n)$ & (\shortthmrefs{lower:tree}{upper:tree})
    \\
    star
    & $\Theta(n \log n)$ & \cite[Cor.~9]{Zu92}
    & $\Theta(n)$ & (\shortthmref{lower:dense})
    \\
    \hline
    cycle
    & $\Theta(n^2)$ & \cite[Ex.~1]{Lovasz93random}
    & $\Theta(n^2)$ & (\shortthmrefs{lower:cycle}{upper:torus})
    \\
    lollipop graph
    & $\Theta(n^3)$ & \cite[Thm.~2.1]{Lovasz93random}
    & $\Theta(n^3)$ & (\shortthmrefs{lower:lollipop}{upper:lollipop})
    \\
    expander
    & $\Theta(n \log n)$ & \cite[Cor.~6]{BK89}, \cite{Rubinfeld90}
    & $\Theta(n \log n)$ & (\shortthmref{lower:expander}, \shortcorref{upper:expander})
    \\
    \hline
    two-dim.~torus
    & $\Theta(n \log^2 n)$ & \cite[Thm.~4]{Zu92}, \cite[Thm.~6.1]{CRRST97}
    & $\Theta(n^{1.5})$ & (\shortthmrefs{lower:torus}{upper:torus})
    \\
    $d$-dim.~torus ($d\geq3$)
    & $\Theta(n \log n)$ & \cite[Cor.~12]{Zu92}, \cite[Thm.~6.1]{CRRST97}
    & $\O(n^{1+1/d})$ & (\shortthmref{upper:torus})
    \\
    hypercube
    & $\Theta(n \log n)$ & \cite[p.~372]{Al83}, \cite[Sec.~5.2]{Palacios94}
    & $\Theta(n \log^2 n)$ & (\shortthmrefs{lower:hypercube}{upper:hypercube})
    \\
    complete
    & $\Theta(n \log n)$ & \cite[Ex.~1]{Lovasz93random}
    & $\Theta(n^2)$ & (\shortthmrefs{lower:dense}{upper:complete})
    \\
    \hline
    \end{tabular}
    \end{center}
    \caption{Comparison of the vertex cover time of random and deterministic walk on different graphs ($n=|V|$).}
    \vspace*{-.5\baselineskip}
    \label{tab:vertex}
\end{table*}

\subsection{Cover time of random walks}

\noindent
In his survey, Lov\'{a}sz~\citep{Lovasz93random} mentions
three important measures of a random walk: cover time,
hitting time, and mixing time.  These three (especially the first two) are closely related,
here we will mainly concentrate on the \emph{cover time}
which is the expected number of steps to visit every node.
The study of the cover time of random walks on graphs was initiated in 1979. Motivated by
the space-complexity of the $s$--$t$-connectivity problem, \citet{AKLLR79}
showed that the cover time is upper bounded by $\Oh(|V| \, |E|)$ for any graph.
For regular graphs,
\citet{Fe97} gave an improved upper bound of $\Oh(|V|^2)$ for the cover time.
\citet{BK89} proved several bounds which rely on the spectral
gap of the transition matrix. Their bounds imply that the cover time on a
regular expander graph is $\Theta(|V| \log |V|)$.
In addition, many papers are
devoted to the study of the cover time on special graphs such as
hypercubes~\cite{Al83},
random graphs~\cite{CF08,CF07,CF05b},
random regular graphs~\cite{CF05},
random geometric graphs~\cite{CF09},
and planar graphs~\cite{JS00}.
A general lower bound of $(1-o(1)) \, |V| \ln |V|$ for any graph was shown by Feige~\cite{Fe95b}.

A natural variant of the cover time is the so-called \emph{edge cover
time}, which measures the expected number of steps to traverse all edges.
Amongst other results, \citet{Zu91,Zu92} proved that the edge cover time
of general graphs is at least $\Omega(|E| \log |E|)$ and at most $\Oh(|V| \, |E|)$.
Finally, \citet{BF93,BF96journal} considered the time until a certain number of vertices (or edges) has been visited.

\begin{table*}[t]
    \footnotesize
    \begin{center}
    \begin{tabular}{|l|cl|cl|}
    \hline
    \multirow{2}{*}{\bf Graph class $\bm G$} &
    \multicolumn{2}{|c|}{\bf Edge cover time $\bmrEC\bm{(G)}$} &
    \multicolumn{2}{|c|}{
        \begin{minipage}[b]{0.001mm}\vspace*{.42cm}\end{minipage}
        \bf Edge cover time $\bmdEC\bm{(G)}$
        \begin{minipage}[b]{0.001mm}\vspace*{.42cm}\end{minipage}
    }
    \\
    &
    \multicolumn{2}{|c|}{\bf of the random walk} &
    \multicolumn{2}{|c|}{\bf of the deterministic walk} \\
    \hline
    \hline
    \kary tree, $k=\Oh(1)$
    & $\Theta(n \log^2 n)$ & \cite[Cor.~9]{Zu92}
    & $\Theta(n \log n)$ & (\shortthmrefs{lower:tree}{upper:tree})
    \\
    star
    & $\Theta(n \log n)$ & \cite[Cor.~9]{Zu92}
    & $\Theta(n)$ &  (\shortthmref{lower:dense})
    \\
    complete
    & $\Theta(n^2 \log n)$ & \cite{Zu91,Zu92} 
    & $\Theta(n^2)$ & (\shortthmrefs{lower:dense}{upper:complete})
    \\
    \hline
    expander
    & $\Theta(n \log n)$ & \cite{Zu91,Zu92}
    & $\Theta(n \log n)$ & (\shortthmref{lower:expander}, \shortcorref{upper:expander})
    \\
    cycle
    & $\Theta(n^2)$ & \cite[Ex.~1]{Lovasz93random}
    & $\Theta(n^2)$ & (\shortthmrefs{lower:cycle}{upper:torus})
    \\
    lollipop graph
    & $\Theta(n^3)$ & \cite[Thm.~2.1]{Lovasz93random}, \cite[Lem.~2]{Zu91} 
    & $\Theta(n^3)$ & (\shortthmrefs{lower:lollipop}{upper:lollipop})
    \\
    hypercube
    & $\Theta(n \log^2 n)$ & \cite{Zu91,Zu92}
    & $\Theta(n \log^2 n)$ & (\shortthmrefs{lower:hypercube}{upper:hypercube})
    \\
    \hline
    two-dim.~torus
    & $\Theta(n \log^2 n)$ & \cite{Zu91,Zu92} 
    & $\Theta(n^{1.5})$ & (\shortthmrefs{lower:torus}{upper:torus})   \\
    $d$-dim.~torus ($d\geq3$)
    & $\Theta(n \log n)$ & \cite{Zu91,Zu92}  
    & $\O(n^{1+1/d})$ & (\shortthmref{upper:torus})
    \\
    \hline
    \end{tabular}
    \end{center}
    \caption{Comparison of the edge cover time of random and deterministic walk on different graphs ($n=|V|$).}
    \vspace*{-.5\baselineskip}
    \label{tab:edge}
\end{table*}

\subsection{Cover time of deterministic walks (our results)}
\label{sec:ourresults}

\noindent
For the case of a cyclic rotor sequence
the edge cover time is known to be $\Theta(|E|\,\diam(G))$ (see \citet{YanovskiWB03} for the upper and \citet{Bampas09} for the lower bound).
It is further known that there are rotor sequences such that
the edge cover time is precisely $|E|$~\citep{Priezzhev1996}.
We allow \emph{arbitrary} rotor sequences and
present three techniques to upper bound the edge cover time 
based on the local divergence~(\shortthmref{divergencecover}),
expansion of the graph~(\shortthmref{upper:expander}),
and a corresponding flow problem~(\shortthmref{flowcovernonregular}).
With these general theorems it is easy to prove upper
bounds for expanders, complete graphs, torus graphs, hypercubes, \kary trees
and lollipop graphs.
Though these bounds are known to be tight, it is illuminating to study
which setup of the rotors matches these upper bounds.
This is the motivation for \secref{lower} which presents matching lower bounds
for all forementioned graphs by describing the precise setup of the rotors.

It is not our aim to prove superiority of the deterministic walk,
but it is instructive to compare our results for the vertex and edge cover time
with the respective bounds of the random walk.
\tabrefs{vertex}{edge} group the graphs in three classes depending whether
random or deterministic walk is faster.
In spite of the strong adversary (as the order of the rotors is completely arbitrary),
the deterministic walk is surprisingly efficient.
It is known that the edge cover time of random walks can be asymptotically
larger than its vertex cover time.  Somewhat unexpectedly, this is not the case for the deterministic walk.
To highlight this issue, let us consider hypercubes and complete graph.
For these graphs, the vertex cover time of the deterministic
walk is larger while the edge cover time is smaller (complete graph) or equal (hypercube)
compared to the random walk.

Analogous to the results of \citet{BF93,BF96journal} for random walks,
we also analyze the short term behavior of the deterministic walk in \secref{shortterm}.
As an example observe that \thmref{shortterm} proves that for $1\leq\alpha<2$
the deterministic walk only needs $\O(|V|^\alpha)$ steps to visit $|V|^\alpha$ edges
of any graph with minimum degree~$\Omega(n)$ while the random walk
needs $\O(|V|^{2\alpha-1})$
steps according to \citep{BF93,BF96journal} (cf.~\tabref{shortterm}).

\section{Models and Preliminaries}

\subsection{Random Walks}
\label{sec:model:rnd}

\noindent
We consider weighted random walks on finite connected graphs $G=(V,E)$. For this, we
assign every pair of vertices $u,v\in V$ a weight $c(u,v) \in \N_{0}$ (rational
weights can be handled by scaling) such that
$c(u,v)=c(v,u)>0$ if $\{u,v\}\in E$ and $c(u,v)=c(v,u)=0$ otherwise. This
defines transition probabilities $\bP_{u,v}:=c(u,v) / c(u)$ with
$c(u):=\sum_{w\in V} c(u,w)$. So, whenever a random walk is at a vertex $u$ it moves
to a vertex $v$ in the next step with probability $\bP_{u,v}$.
Moreover, note that for all $u,v\in V$, $c(u,v)=c(v,u)$
while $\bP_{u,v}\neq\bP_{v,u}$ in general. This defines a time-reversible,
irreducible, finite Markov chain $X_0,X_1,\ldots$ with transition matrix~$\bP$ (cf.~\cite{AF02}).
The $t$-step probabilities of the walk can be obtained by taking the $t$-th
power of~$\bP^t$.
In what follows, we prefer to use the term weighted
random walk instead of Markov chain to emphasize the limitation to rational
transition probabilities.

It is intuitively clear that a random walk with large weights~$c(u,v)$
is harder
to approximate deterministically with a simple rotor sequence.
To measure this, we use $c_{\max}:=\max_{u,v\in V} c(u,v)$.
An important special case is the \emph{unweighted random walk}
with $c(u,v)\in\{0,1\}$ for all $u,v\in V$ on a simple graph.
In this case, $\bP_{u,v}=1/\deg(u)$ for all $\{u,v\} \in E$, and $c_{\max}=1$.
Our general results hold for weighted (random) walks.  However,
the derived bounds for specific graphs are only stated for unweighted walks.
With \emph{random walk} we mean unweighted random walk and
if a random walk is allowed to be weighted we will emphasize this.

For weighted and unweighted random walks we define for a graph~$G$,
\vspace*{-.3\baselineskip}
\begin{list}{\labelitemi}{\leftmargin=1.81em}  
    \setlength{\itemsep}{0pt}
    \setlength{\parskip}{0pt}
    \item cover time:
        $\rVC(G) =
            \max_{u\in V}
            \Ex{ \min
                     \big\{ t\geq0 \colon \textstyle\bigcup_{\ell=0}^t \{ X_\ell \} = V \big\}\mid X_0 = u }$,
    \item edge cover time:
        $\rEC(G) =
            \max_{u\in V} \Ex{ \min
             \big\{ t\geq0 \colon \textstyle\bigcup_{\ell=1}^t \{ X_{\ell-1},X_{\ell} \} = E\big\}\mid X_0 = u }$.
\end{list}
The (edge) cover time of a graph class~$\cG$ is
the maximum of the (edge) cover times of all graphs of the graph class.
Observe that $\rVC(\cG)\leq\rEC(G)$ for all graphs~$G$.
For vertices $u,v\in V$ we further define
\vspace*{-.3\baselineskip}
\begin{list}{\labelitemi}{\leftmargin=1.81em}  
    \setlength{\itemsep}{0pt}
    \setlength{\parskip}{0pt}
    \item (expected) hitting time:
        $\rH(u,v) = \Ex{ \min \left\{ t \geq 0 \colon X_{t} = v \right\} \mid X_0 = u } $,
    \item stationary distribution:
        $\pi_u = c(u) / \sum_{w\in V} c(w)$.
\end{list}

\subsection{Deterministic Random Walks}
\label{sec:model:det}

\noindent
We define weighted deterministic random walks (or short:
weighted deterministic walks) based on rotor routers as introduced by \citet{HP09}.
For a weighted random walk, we
define the corresponding weighted deterministic walk as follows.
We use a tilde (\:$\widetilde\ $\:) to mark variables
related to the deterministic walk.
To each vertex~$u$ we assign a rotor sequence
$\ts(u)=(\ts(u,1),\ts(u,2),\ldots,\ts(u,\td(u)))\in V^{\td(u)}$
of arbitrary length~$\td(u)$ such that
the number of times a neighbor~$v$ occurs in the rotor sequence~$\ts(u)$
corresponds to the transition probability to go from~$u$ to~$v$ in the weighted random walk,
that is,
$
    \bP_{u,v} = | \{ i\in[\td(u)] \colon \ts(u,i)=v \} | / \td(u)
$
with $[\td(u)]:=\{1,\ldots,\td(u)\}$.
For a weighted random walk,~$\td(u)$ is a multiple
of the lowest common denominator of the transition probabilities from~$u$
to its neighbors.
For the standard random walk, a corresponding canonical deterministic walk
would be $\td(u)=\deg(u)$ and
a permutation of the neighbors of~$u$ as rotor sequence~$\ts(u)$.
As the length of the rotor sequences crucially influences the performance
of a deterministic walk, we set $\tkappa:=\max_{u\in V} \td(u) / \deg(u)$ (note that $\tkappa \geq 1$).
The set~$V$ together with~$\ts(u)$ and~$\td(u)$ for all $u\in V$
defines the deterministic walk, sometimes abbreviated~$\bD$.
Note that every deterministic walk has a unique corresponding random walk
while there are many deterministic walks corresponding to one random walk.

We also assign to each vertex~$u$ an integer $\tr_t(u)\in[\td(u)]$ corresponding
to a rotor at~$u$ pointing to~$\ts(u,\tr_t(u))$ at step $t$.
A \emph{rotor configuration}~$C$ describes the rotor sequences~$\ts(u)$ and
initial rotor directions~$\tr_0(u)$ for all vertices $u\in V$.
At every time step~$t$ the walk moves from
$\tx_t$ in the direction of the current rotor of~$\tx_t$ and this rotor is
incremented\footnote{In this respect we slightly deviate from the model of \citet{HP09}
who first increment the rotor and then move the chip, but this change
is insignificant here.}
to the next position according to the rotor sequence~$\ts(\tx_t)$ of~$\tx_t$.
More formally, for given~$\tx_t$ and $\tr_t(\cdot)$ at time~$t\geq0$
we set
$\tx_{t+1}:=s(\tx_t,\tr_t(\tx_t))$,
$\tr_{t+1}(\tx_t):=\tr_t(\tx_t) \mod \td(\tx_t)+1$, and
$\tr_{t+1}(u):=\tr_{t}(u)$ for all $u\neq \tx_t$.
Let~$\calC$ be the set of all possible rotor configurations (that is,
$\ts(u)$, $\tr_0(u)$ for $u\in V$)
of a corresponding deterministic walk for a fixed weighted random walk (and fixed rotor sequence length $\td(u)$ for each $u \in V$).
Given a rotor configuration $C\in \calC$ and an initial location $\tx_0\in V$,
the vertices $\tx_0,\tx_1,\ldots\in V$ visited by a deterministic walk
are completely determined.

For deterministic walks we define for a graph~$G$ and vertices $u,v\in V$,
\vspace*{-.3\baselineskip}
\begin{list}{\labelitemi}{\leftmargin=1.81em}  
    \setlength{\itemsep}{0pt}
    \setlength{\parskip}{0pt}
    \item deterministic cover time:
        $\dVC(G) = \max_{\tx_0\in V} \max_{C\in \calC} \min
                 \big\{ t\geq0 \colon \textstyle\bigcup_{\ell=0}^t \{ \tx_\ell \} = V\big\}$,
    \item deterministic edge cover time:\\
        $\dEC(G) = \max_{\tx_0\in V} \max_{C\in \calC} \min
             \big\{ t\geq0 \colon \textstyle\bigcup_{\ell=1}^t \{ \tx_{\ell-1},\tx_{\ell}\} = E\big\}$,
    \item hitting time:
        $\dH(u,v) = \max_{C\in \calC} \min \left\{ t \geq 0 \colon \tx_{t} = u, \tx_0 = v  \right\}$.
\end{list}
Note that the definition of the deterministic cover time
takes the \emph{maximum} over all possible rotor configurations, while
the cover time of a random walk takes the \emph{expectation} over the random decisions.
Also, $\dVC(G)\leq\dEC(G)$ for all graphs~$G$.
We further define
for fixed configurations $C\in\calC$, $\tx_0$, and vertices $u,v\in V$,
\vspace*{-.3\baselineskip}
\begin{list}{\labelitemi}{\leftmargin=1.81em}  
    \setlength{\itemsep}{0pt}
    \setlength{\parskip}{0pt}
    \item number of visits to vertex $u$: 
        $\tN_t(u) = \big| \{ 0\leq \ell \leq t \colon \tx_\ell = u \} \big|$,
    \item number of traversals of a directed edge $u\to v$: \\
        $\tN_t(u\to v) = \big| \{ 1\leq \ell \leq t \colon (\tx_{\ell-1},\tx_\ell) = (u,v) \} \big|$.
\end{list}

\subsection{Graph-Theoretic Notation}

\noindent
We consider finite, connected graphs $G=(V,E)$.
Unless stated differently, $n:=|V|$ is the number vertices and
$m:=|E|$ the number of undirected edges.
By~$\delta$ and~$\Delta$ we denote the minimum and maximum degree of the graph, respectively.
For a pair of vertices $u,v \in V$,
we denote by $\dist(u,v)$ their distance, i.e., the length of a shortest path between them.
For a vertex $u \in V$, let $\Gamma(u)$ denote the set of all neighbors of~$u$.
More generally, for any $k \geq 1$,
$\Gamma^k(u)$ denotes the set of vertices~$v$ with $\dist(u,v)=k$.
For any subsets $S, T \subseteq V$, $E(S)$ denotes the set of edges with one endpoint in~$S$ and
$E(S,T)$ denotes the edges~$\{u,v\}$ with $u\in S$ and  $v\in T$.
As a walk is something directed, we also have to argue about directed edges
though our graph~$G$ is undirected.  In slight abuse of notation,
for~$\{u,v\} \in E$ we might also write $(u,v)\in E$ or $(v,u)\in E$.
Finally, all logarithms used here are to the base of~$2$.

\section{Upper Bounds on the Deterministic Cover Times}
\label{sec:upper}

\noindent
Very recently, \citet{HP09} proved that several natural quantities of the
weighted deterministic walk as defined in \secref{model:det}
concentrate around the respective expected values
of the corresponding weighted random walk.
To state their result formally, we set for a vertex $v\in V$,
\begin{equation}
\label{eq:K}
    K(v) := \max_{u \in V} \rH(u,v) +
         \frac{1}{2} \biggl( \frac{\td(v)}{\pi_v} +
         \sum_{i,j \in V} \td(i) \,\bP_{i,j}
            \left| \rH(i,v) - \rH(j,v) - 1 \right| \biggr).
\end{equation}

\begin{thm}[{\cite[Thm.~4]{HP09}}]
    \label{thm:propp}
    For 
    all weighted deterministic walks,
    all vertices~$v \in V$,
    and all times~$t$,
    \[
        \bigg| \pi_v - \frac{\tN_t(v)}t \bigg|
        \leq \frac{K(v)\,\pi_v}t.
    \]
\end{thm}

\noindent
Roughly speaking, \thmref{propp} states that the proportion
of time spent by the weighted deterministic walk concentrates around the
stationary distribution for all configurations $C\in\calC$
and all starting points~$\tx_0$.
To quantify the hitting or cover time with \thmref{propp},
we choose $t=K(v)+1$ to get $\tN_t(v) > 0$.  To get a bound for the edge cover time,
we choose $t=3K(v)$ and observe that then $\tN_t(v) \geq 2 \pi_v K(v) > \td(v)$.
This already shows the following corollary.

\begin{cor}\label{cor:propp}
    For 
    all weighted deterministic walks,
    \begin{align*}
        \dH(u,v)
        &\leq
        K(v)+1 \qquad\qquad\qquad\text{for all~$u,v\in V$,}
        \\
        \dVC(G)
        &\leq
        \max_{v\in V}K(v)+1,
        \\
        \dEC(G)
        &\leq
        3\max_{v\in V}K(v).
    \end{align*}
\end{cor}

\noindent
One obvious question that arises from \thmref{propp} and \mycorref{propp}
is how to bound the value~$K(v)$. While it is clear that~$K(v)$ is
polynomial in~$n$ (provided that~$c_{\max}$ and~$\tkappa$ are polynomially bounded), it is not clear how to get more precise upper bounds.
A key tool to tackle the difference of hitting times in~$K(v)$
is the following elementary lemma,
where in case of a periodic walk
the sum is taken as a Ces{\'a}ro summation~\citep{Cesaro}.

\newcommand{\texttriple}{
    For all weighted random walks and all vertices $i,j,v\in V$,
    \[
      \tsum_{t=0}^{\infty} \big( \bP_{i,v}^{t} - \bP_{j,v}^{t} \big)
      =
      \pi_v \, ( \rH(j,v)-\rH(i,v) ).
    \]
}
\mylem{triple}{\texttriple}
\begin{proof}
    Let $\bZ$ be the fundamental matrix of~$\bP$ defined
    as
    $
      \bZ_{ij} := \sum_{t=0}^{\infty} \big( \bP_{i,j}^{t} - \pi_j \big).
    $
    It is known that for any pair of vertices~$i$ and~$v$,
    $
    \pi_v \, \rH(i,v) = Z_{v v} - Z_{i v}$
    (cf.~\cite[Ch.~2, Lem.~12]{AF02}). Hence by the convergence of~$\bP$,
    \begin{align*}
     \pi_v (\rH(j,v) - \rH(i,v))
     &= (Z_{v v} - Z_{j v}) - (  Z_{v v} - Z_{i v} ) \\
     &=  \tsum_{t=0}^{\infty} \big( \bP_{i,v}^{t} - \pi_v  \big)
        - \tsum_{t=0}^{\infty} \big( \bP_{j,v}^{t} - \pi_v  \big)
        = \tsum_{t=0}^{\infty} \big(	\bP_{i,v}^{t} - \bP_{j,v}^{t}	\big).
        \qedhere
    \end{align*}
\end{proof}

\subsection{Bounding $K(v)$ by the local divergence}

\noindent
To analyze weighted random walks, we use the notion of local divergence which
has been a fundamental quantity in the analysis of load balancing algorithms~\cite{RSW98,FS09}.
Moreover, the local divergence is considered to be of
independent interest (see \cite{RSW98} and further references therein).

\begin{defi} 
    \label{def:divergence}
    The local divergence of a weighted random walk is
    $
     \Psi(\bP) :=  \max_{v \in V} \Psi(\bP,v),
    $
    where $\Psi(\bP,v)$ is the local divergence
    w.r.t.\ to a vertex $v \in V$ defined as
    $
        \Psi(\bP,v) :=
        \tsum_{t=0}^{\infty}
        \tsum_{ \{ i,j \} \in E}
        \big| \bP_{i,v}^{t} - \bP_{j,v}^{t}  \big|.
    $
\end{defi}

Using \mycorref{propp} and \lemref{triple}, we get the following bound
on the hitting time of a deterministic walk.

\newcommand{\textdivergencecover}{
    For 
    all deterministic walks
    and all vertices~$v\in V$,
    \[
    K(v) \leq \max_{u \in V} \rH(u,v)
            + \frac{\tkappa\,c_{\max}}{\pi_v}\Psi(\bP,v)  + 2m\,\tkappa \,c_{\max}.
    \]
}
\mythm{divergencecover}{\textdivergencecover}
\begin{proof}
    To bound~$K(v)$ we first observe that by definition of~$\tkappa$ and~$c_{\max}$
    for all $u,v\in V$,
    \[
        \frac{\td(v)}{\pi_v}
        =
        \frac{\td(v)\,\sum_{i,j\in V}c(i,j)}{c(v)}
        \leq
        \frac{\tkappa\deg(v)\,2\sum_{\{i,j\}\in E}c(i,j)}{c(v)}
        \leq
        2\tkappa\tsum_{\{i,j\}\in E}c(i,j)
        \leq
        2m \,\tkappa\,c_{\max},
    \]
    \[
        \td(u) \bP_{u,v}
        \leq
        \tkappa\,\deg(u)\,\bP_{u,v}
        =
        \frac{\tkappa\,\deg(u)\,c(u,v) }{ c(u)}
        \leq
        \tkappa\,c(u,v)
        \leq
        \tkappa\,c_{\max}.
    \]
    Therefore,
    \begin{align*}
        K(v) &\leq \max_{u \in V} \rH(u,v)
               +m\,\tkappa\,c_{\max}
               +\frac{1}{2} \sum_{i,j \in V}
               \tkappa\,c_{\max} \, \big(| \rH(i,v) - \rH(j,v) | + 1 \big)  \\
       &\leq \max_{u \in V}  \rH(u,v)
            +2m\,\tkappa\,c_{\max}
            + \tkappa\,c_{\max} \sum_{ \{i,j\} \in E} |\rH(i,v) - \rH(j,v)| \\
       &\leq \max_{u \in V} \rH(u,v)
            + 2m\,\tkappa \,c_{\max}
            + \frac{\tkappa\,c_{\max}}{\pi_v}\Psi(\bP,v),
    \end{align*}
    where the last inequality follows from \lemref{triple} and \defref{divergence}.
\end{proof}

To see where the dependence on $\tkappa$ in \thmref{divergencecover}
comes from,
remember that our bounds hold for all configurations~$C\in\calC$ of the deterministic walk.
This is equivalent to bounds for a walk where an adversary chooses the
rotor sequences within the given setting.  Hence a larger~$\tkappa$ strengthens
the adversary as it gets more freedom of choice in the order
of the rotor sequence.
On the other hand, the~$c_{\max}$ measures how skewed the probability distribution
of the random walk can be.  With larger~$c_{\max}$, they get harder to approximate
deterministically.

Note that \thmref{divergencecover} is more general than just giving
an upper bound for hitting and cover times via \mycorref{propp}.
It can be useful in the other directions, too.
To give a specific example, we can apply the result of
\thmref{lower:hypercube} that $\dEC(G)=\Omega(n \log^2 n)$ for hypercubes
and $\max_{u,v} \rH(u,v) = \Oh(n)$ (cf.~\cite{Lovasz93random}) to
\thmref{divergencecover} and obtain a lower bound of $\Omega(n \log^2 n)$
on the local divergence of hypercubes.

\subsection{Bounding $K(v)$ for symmetric walks}

\noindent
To get meaningful bounds for the cover time, we restrict to unweighted random walks in the following.
In our notation this implies~$c_{\max}=1$ while~$\tkappa$ is still arbitrary.
First, we derive a tighter version of \thmref{divergencecover} for symmetric
random walks defined as follows.

\begin{defi}
    \label{def:symP}
    A symmetric random walk has transition probabilities
    $\bP'_{u,v} = \frac{1}{\Delta+1}$ if $\{u,v\} \in E$,
    $\bP'_{u,u} = 1 - \frac{1}{\Delta+1} \deg(u)$ and
    $\bP'_{u,v} = 0$ otherwise.
\end{defi}

These symmetric random walks occur frequently in the literature, e.g.,
for load balancing \cite{RSW98,FS09}
or for the cover time \cite{AKL08}.
The corresponding deterministic walk is defined as follows.

\begin{defi}
    \label{def:symdet}
    For an unweighted deterministic walk~$\bD$ with rotor sequences
    $\ts(\cdot)$ of length~$\td(\cdot)$, let the corresponding
    symmetric deterministic walk~$\bD'$ have for all~$u\in V$
    rotor sequences
    $\ts\mspace{1.5mu}'(u)$ of length $\td'(u):=\frac{\Delta+1}{\deg(u)}\td(u)$.
    with $\ts\mspace{1.5mu}'(u,i):=\ts(u,i)$ for $i\leq\td(u)$ and
    $\ts'(u,i):=u$ for $i>\td(u)$.
\end{defi}

\begin{wrapfigure}{r}{1.8cm}
  \vspace*{-2.1\baselineskip}
  \begin{center}
    \begin{tikzpicture}[->,shorten >=1pt,auto,thick,
    every node/.style={text centered,font=\normalsize}]
    \node (RW) at (0,.9) {$\bP$};
    \node (dRW) at (1.1,.9) {$\bD$};
    \node (sRW) at (0,0) {$\bP'$};
    \node (sdRW) at (1.1,0) {$\bD'$};
    \draw[thick] (node cs:name=dRW) -- (node cs:name=sdRW); 
    \draw[thick] (node cs:name=RW) -- (node cs:name=sRW); 
    \draw[thick] (node cs:name=sRW) -- (node cs:name=sdRW); 
    \draw[thick] (node cs:name=RW) -- (node cs:name=dRW);
    \end{tikzpicture}
  \end{center}
  \vspace*{-1\baselineskip}
\end{wrapfigure}

It is easy to verify that the definition ``commutes'', that is,
for a deterministic walk~$\bD$ corresponding to a random walk~$\bP$,
the corresponding deterministic walk~$\bD'$ corresponds to the
corresponding symmetric random walk~$\bP'$.
Let all primed variables
($\pi'_u$, $K'(v)$, $\kappa'$, $c'(u,v)$, $c_{\max}'$, $\rH'(u,v)$, $\dH'(u,v)$, $\rVC'(G)$, $\dVC'(G)$, $\rEC'(G)$, $\dEC'(G)$)
have their natural meaning for the symmetric random walk
and symmetric deterministic walk.

As~$\bP'$ is symmetric, the stationary distribution of~$\bP'$ is uniform,
i.e., $\pi'_i=1/n$ for all $i\in V$.
Note that the symmetric walk is in fact a weighted walk with
$c'(u,v)=1$ for $\{u,v\}\in E$,
$c'(u,u)=\Delta+1-\deg(u)$ for $u\in V$,
and
$c'(u,v)=0$ otherwise.
Using $c_{\max}'=\Delta+1-\delta$ in \thmref{divergencecover} is too coarse.
To get a better bound on~$K'(v)$ for symmetric walks, observe
that for all $v\in V$
    \begin{equation}
    \label{eq:symdpi}
        \frac{\td'(v)}{\pi'(v)}
        =
        n \td(v) \frac{\Delta+1}{\deg(v)}
        \leq
        n \,\tkappa \,(\Delta+1)
    \end{equation}
and for all $\{u,v\}\in E$
    \begin{equation}
        \label{eq:symdP}
        \td'(u) \bP'_{u,v}
        =
        \frac{\td(u)}{\deg(u)}
        \leq
        \tkappa.
    \end{equation}
Plugging this in the definition of~$K(v)$ as in \thmref{divergencecover} gives
the following theorem.

\begin{thm}\label{thm:symdivergence}
    For all symmetric 
    deterministic walks
    and all vertices $v\in V$,
    \[
    K'(v) =\Oh\bigg( \max_{u \in V} \rH'(u,v)
            + \frac{\tkappa}{\pi'(v)}\Psi(\bP',v)  + n\,\Delta\,\tkappa \bigg).
    \]
\end{thm}

\noindent
By definition,
$\dEC(G) \leq \dEC'(G)$
and
$\rH(u,v) \leq \rH'(u,v)$ for all $u,v\in V$.
The following lemma gives a natural reverse of the latter inequality.

\begin{lem}
    \label{lem:symH}
    For a random walk~$\bP$ and a symmetric random walk~$\bP'$ it holds for any pair of vertices $u,v$ that
    \[
       \rH'(u,v)
       \leq
       \frac{\Delta+1}{\delta} \, \rH(u,v).
    \]
\end{lem}
\begin{proof}
    Let us consider the transition matrix~$\bP''$ with $\bP''_{u,u} = 1
    - \frac{\delta}{\Delta+1}$, $\bP''_{u,v} = \frac{\delta}{\Delta+1} \cdot
    \frac{1}{\deg(u)} $ if $\{u,v\} \in E$ and $\bP''_{v,v} = 0$ otherwise. Let
    $\rH''$ denote the hitting times of a random walk according to~$\bP''$.
    We couple the non-loop steps of a random walk according to~$\bP'$
    with the non-loop steps of a random walk according to~$\bP''$, as in both walks, a neighbor is chosen uniformly at random (conditioned on the event that the walk does not loop).

     Since all
    respective loop-probabilities satisfy $\bP'_{u,u} \leq \bP''_{u,u}$, it
    follows that for all vertices $u,v \in V$,
    $
      \rH'(u,v) \leq \rH''(u,v).
    $
    Our next aim is to relate $\tau''(u,v)$ to~$\tau(u,v)$, where $\tau''$ ($\tau$, resp.) is the first step when a random walk according to~$\bP''$ ($\bP$, resp.) starting at $u$ visits $v$. We can again couple the non-loop steps of both random walks, since every non-loop step of $\bP''$ chooses a uniform neighbor and so does $\bP$. Hence,
    $
     \rH''(u,v)  = \Ex{ \tau''(u,v) }
                 = \Exbig{ \sum_{i=1}^{\tau(u,v)} X_i},
    $
    where the $X_i$'s are independent, identically distributed geometric random variable with mean $\frac{\Delta+1}{\delta}$.
    Applying Wald's equation~\cite{Wald44}
     yields
    \begin{align*}
      \rH''(u,v) &= \Ex{\tau(u,v)} \cdot \Ex{X_1} = \rH(u,v) \cdot \frac{\Delta+1}{\delta},
    \end{align*}
    which  proves the claim.
\end{proof}

\subsection{Upper bound on the deterministic cover time depending on the expansion}

\noindent
We now derive an upper bound for~$\dEC(G)$ that depends on the expansion properties
of $G$.
Let~$\lambda_2(\bP)$ be the second-largest eigenvalue in absolute value of~$\bP$.

\newcommand{\textupperexpander}{
    For all graphs~$G$,
    $\dEC(G) = \Oh \big( \frac{\Delta}{\delta} \, \frac{n}{1-\lambda_2(\bP)} +
        n \, \tkappa\, \frac{\Delta}{\delta} \, \frac{\Delta \log n }{1-\lambda_2(\bP)} \big)
        $.
}
\mythm{upper:expander}{\textupperexpander}
\begin{proof}
    Let~$\bP$ and~$\bD$ be corresponding unweighted random and deterministic walks
    and~$\bP'$ and~$\bD'$ be defined as in \defrefs{symP}{symdet}.
    From the latter definition we get $\dEC(G)\leq \dEC'(G)$, as additional loops in the rotor sequence can only slow down the covering process.
    Hence it suffices to bound $\dEC'(G)$ with \thmref{symdivergence}.
    We will now upper bound all three summands involved in \thmref{symdivergence}.

    By two classical result for reversible, ergodic Markov chains (\cite[Chap.~3, Lem.~15]{AF02}
    and \cite[Chap.~3, Lem.~17]{AF02} of \citeauthor{AF02}),
    \begin{align*}
     \max_{u,v} \rH'(u,v)
     \leq
     2 \sum_{u \in V} \pi_u \cdot \rH'(u,v)
     \leq
     2 \frac{1-\pi_v}{\pi_v \cdot (1-\lambda_2(\bP'))}.
    \end{align*}
    As~$\bP'$ is symmetric, the stationary distribution of~$\bP'$ is uniform and therefore
    \begin{align}
        \max_{u,v\in V} \rH'(u,v) &\leq 2 \frac{n}{1-\lambda_2(\bP')}. \label{eq:maxhit}
    \end{align}

    \noindent
    In order to relate $\lambda_2(\bP)$ and $\lambda_2(\bP')$,
    we use the following ``direct comparison lemma'' for reversible Markov Chains~$\bP$ and~$\bP'$
    from \cite[Eq.~2.3]{DS93} (where in their notation, we plug in
    $a= \min_{i\in V} \frac{\pi_i}{\pi'_i} $ and $A=\max_{(i,j) \in E, i \neq j} \frac{\pi_i \bP_{i,j}}{\pi'_j \bP'_{i,j}}  $) to obtain that
    \begin{align}
        \label{eq:DS93}
        \frac{1-\lambda_2(\bP)}{1-\lambda_2(\bP')} \,
        \leq \,
        \frac{\max_{ (i,j) \in E, i \neq j} \frac{\pi_i \bP_{i,j}}{\pi'_i \bP'_{i,j} } }
           { \min_{i\in V} \frac{\pi_i}{\pi'_i} }.
    \end{align}
    We now determine the denominator and numerator of the right hand side of \eq{DS93}.
    As $\pi'_i=1/n$ and $\pi_i = \frac{\deg(i)}{2m}$ for all $i \in V$,
    $\min_i \frac{\pi_i}{\pi'_i} = \frac{\delta}{2m} n$.
    Moreover, for any edge $\{i,j\} \in E$,
    $\pi_i \bP_{i,j} = \frac{\deg(i)}{2m} \,\frac{1}{\deg(i)} = \frac{1}{2m}$ and
    $\pi'_i \bP'_{i,j} = \frac{1}{n} \, \frac{1}{\Delta+1}$ and therefore
    $
    \max_{(i,j) \in E, i \neq j} \frac{\pi_i \bP_{i,j}}{\pi'_i \bP'_{i,j} } =
    \frac{n (\Delta+1)}{2m}
    $.
    Plugging this  in \eq{DS93} yields
    \begin{align}
    \label{eq:standardzwei}
     \frac{1-\lambda_2(\bP)}{1-\lambda_2(\bP')}
     \leq \frac{ \frac{n(\Delta+1)}{2m} }{\frac{\delta}{2m} \, n}
      = \frac{\Delta+1}{\delta}.
    \end{align}

    \noindent
    From Theorem~4 of \citet{RSW98} we know the
    following upper bound on $\Psi(\bP')$ ,
    \begin{align}
      \Psi(\bP') &= \Oh \left( \frac{\Delta \log n}{1-\lambda_2(\bP')} \right). \label{eq:divergence}
    \end{align}
    Plugging all this in \thmref{symdivergence} and \mycorref{propp} gives
    \begin{align*}
    \dEC'(G)
     &= \O\bigg( \max_{u \in V} \rH'(u,v)
            + \frac{\tkappa}{\pi'(v)}\Psi(\bP',v)
            + n\,\Delta\,\tkappa
            \bigg)
          \\
     &= \O\bigg( \frac{n}{1-\lambda_2(\bP')}
            + n \, \frac{ \Delta \tkappa \log n}{1-\lambda_2(\bP')}
            \bigg)
         & (\text{by \eqs{maxhit}{divergence}})
          \\
     &= \O\bigg( \frac{\Delta}{\delta} \frac{n}{1-\lambda_2(\bP)}
            + n \, \frac{\Delta}{\delta} \frac{ \Delta \tkappa \log n}{1-\lambda_2(\bP)}
            \bigg).
         & (\text{by \eq{standardzwei}})
    \end{align*}
    As $\dEC(G)\leq \dEC'(G)$, this finishes the proof.
\end{proof}

Here, we call a graph with constant maximum degree an expander graph, if $1/(1-
\lambda_2(\bP))=\Oh(1)$ (equivalently, we have for all subsets $X \subseteq V, 1 \leq
|X| \leq n/2$, $|E(X,X^c)| = \Omega(|X|)$~(cf.~\cite[Prop.~6]{DS93})). Using
\thmref{upper:expander}, we immediately get the following upper bound on
$\dEC(G)$ for expanders.

\begin{cor}
\label{cor:upper:expander}
    For all expander graphs, $\dEC(G) = \Oh(\tkappa \, n \, \log n)$.
\end{cor}

\subsection{Upper bound on the deterministic cover time by flows}

\noindent
We relate the edge cover time of the unweighted random walk to the optimal solution of the following flow problem.

\begin{defi}[cmp.~{\cite[Def.~1, Rem.~1]{MS06}}]\label{def:flow}
    Consider the flow problem where a distinguished source node~$s$ sends a flow amount of~$1$ to each other node in the graph.  Then $f_{s}(i,j)$ denotes the load transferred along edge $\{i,j\}$ (note $f_s(i,j)=-f_s(j,i)$) such that $\sum_{\{i,j\} \in E} f_s(i,j)^2$ is minimized.
 \end{defi}

\newcommand{\textflowcovernonregular}{
    For all graphs~$G$,
    \[
    \dEC(G) = \Oh\bigg(
    	   \frac{\Delta}{\delta} \max_{u,v \in V} \rH(u,v) +
    	   \Delta \,n\,\tkappa  +
    	   \tkappa \,
    	   \Delta \max_{s \in V}
    	   \sum_{\{i,j\} \in E}
    	   \left| f_s(i,j) \right|
    	   \bigg)
    \]
    where~$f_s$ is the flow with source~$s$ according to \defref{flow}.
}
\mythm{flowcovernonregular}{\textflowcovernonregular}
\begin{proof}
    Let $\bP$ and $\bD$ be corresponding unweighted random and deterministic walks
    and $\bP'$ and $\bD'$ be defined as in \defrefs{symP}{symdet}.

    By combining the equalities from \cite[Def.~1 \& Thm.~1]{MS06}
    (where we set the flow amount sent by $s$ to any other vertex to~$1$),
       \begin{equation}
        \left| f_s(i,j) \right|
        = \frac{n}{\Delta+1} \,
          \bigg| \sum_{t=0}^{\infty}
          \big({\bP'}_{i,s}^t - {\bP'}_{j,s}^t \big) \bigg|
        \qquad
        \text{for any edge $\{i,j\} \in E$.}
        \label{eq:meyerhenke}
    \end{equation}

    \noindent
    Now plugging \eqs{symdpi}{symdP} in the definition of $K'(v)$ from \eq{K} gives
    \begin{align*}
      K'(v)
    	&= \Oh\bigg(
    	   \max_{u,v \in V} \rH'(u,v)  +
    	   \Delta \,n\,\tkappa  +
    	   \tkappa
    	   \max_{v \in V}
    	   \sum_{\{i,j\} \in E} |\rH'(i,v) - \rH'(j,v)|
    	   \bigg)\\
    	&= \Oh\bigg(
    	   \max_{u,v \in V} \rH'(u,v) +
    	   \Delta \,n\,\tkappa  +
    	   \tkappa
    	   \max_{v \in V}
    	   n \sum_{\{i,j\} \in E}
    	   \left| \sum_{t=0}^{\infty} {\bP'}_{iv}^{t} - {\bP'}_{jv}^t \right|
    	   \bigg)
    	   & \text{(by \lemref{triple})}  \\
    	&= \Oh\bigg(
    	   \max_{u,v \in V} \rH'(u,v)  +
    	   \Delta \,n\,\tkappa  +
    	   \tkappa \,
    	   \Delta \max_{s \in V}
    	   \sum_{\{i,j\} \in E}
    	   \left| f_s(i,j) \right|
    	   \bigg)
    	   & \text{(by \eq{meyerhenke})}  \\
    	&= \Oh\bigg(
    	   \frac{\Delta}{\delta} \max_{u,v \in V} \rH(u,v) +
    	   \Delta \,n\,\tkappa  +
    	   \tkappa \,
    	   \Delta \max_{s \in V}
    	   \sum_{\{i,j\} \in E}
    	   \left| f_s(i,j) \right|
    	   \bigg)
    	   & \text{(by \lemref{symH})}
    \end{align*}

    \noindent
    With \mycorref{propp},
    $\dEC(G) \leq \dEC'(G)\leq 3\max_{v\in V}K'(v)$ finishes the proof.
\end{proof}

\subsection{Upper bounds on the deterministic cover time for common graphs}
\label{sec:upper:cover}

\noindent
We now demonstrate how to apply above general results
to obtain upper bounds for the edge cover time
of the deterministic walk for many common graphs.
As the general bounds \thmrefss{divergencecover}{upper:expander}{flowcovernonregular}
all have a linear dependency on~$\tkappa$, the following upper bounds
can be also stated depending on~$\tkappa$.  However, for clarity
we assume $\tkappa=\O(1)$ here.

\newcommand{\textuppercomplete}{
    For complete graphs, $\dEC(G)=\Oh(n^2)$.
}
\mythm{upper:complete}{\textuppercomplete}
\begin{proof}
    To bound the local divergence $\Psi(\bP')$, observe that
    for any $t \geq 1$, ${\bP'}_{i,j}^t = 1/n$ for every pair $i,j$.
    Hence we obtain
    \[
        \Psi(\bP')
        = \max_{v \in V} \Bigg\{ \sum_{t=0}^{\infty}
          \sum_{\{i,j\} \in E} \left| \bP_{v,i}^{t} - \bP_{v,j}^t \right| \Bigg \}
        = \max_{v \in V} \Bigg \{ \sum_{\{i,j\} \in E} \left| \bP_{v,i}^0 - \bP_{v,j}^0 \right|
          \Bigg \}
        = n-1.
    \]
    Plugging this into \thmref{symdivergence} yields the claim.
\end{proof}

\newcommand{\textuppertorus}{
    For $d$-dimensional torus graphs ($d\geq1$ constant), $\dEC(G)=\Oh(n^{1+1/d})$.
}
\mythm{upper:torus}{\textuppertorus}
\begin{proof}
    Also here, we apply~\thmref{divergencecover} and use
    the bound from \cite[Thm.~8]{RSW98} that $\Psi(\bP)=\Oh(n^{1/d})$.
    It is known that for $d=1$, $\max_{u,v \in V} \rH(u,v) = \Theta(n^2)$,
    $d=2$, $\max_{u,v \in V} \rH(u,v) = \Theta(n \log n)$ and
    for $d\geq3$, $\max_{u,v \in V} \rH(u,v) = \Theta(n)$~(e.g.,~\cite{CRRST97}).
    Hence the claim follows by \thmref{divergencecover}.
\end{proof}

\newcommand{\textupperhypercube}{
    For hypercubes, $\dEC(G) = \Oh(n \log^2 n)$.  
}
\mythm{upper:hypercube}{\textupperhypercube}
\begin{proof}
    To apply \thmref{flowcovernonregular}, we use the strong symmetry of the hypercube $H=H_{\log n}$. More precisely, we use the distance transitivity of the hypercube (cf.~\cite{Bi93}), that is,
    for all vertices $w,x,y,z\in V$
    with $\dist(w,x)=\dist(y,z)$ there is a permutation $\sigma \colon V \rightarrow V$
    with $\sigma(w)=y$, $\sigma(x)=z$ and for all $u,v \in V$,
    $
      \{ u,v \} \in E \, \Leftrightarrow \{ \sigma(u), \sigma(v) \} \in E.
    $

    We proceed to upper bound $\sum_{\{i,j\} \in E} |f_s(i,j)|$, where $f_s$ is defined as in\defref{flow}.
    As one might expect, for distance-transitive graphs the $\ell_2$-minimal flow
    is distributing the flow uniformly among all edges connecting pairs of vertices
    with a different distance to~$s$. More formally, \cite[Thm.~5]{MS06} showed
    that for any two vertices $i,j\in V$ with $i \in \Gamma^d(s)$ and $j \in \Gamma^{d+1}(s)$,
    \[
        \left| f_s(i,j) \right| =
        \frac{1}{|E(\Gamma^d(s),\Gamma^{d+1}(s))|} \sum_{\ell=d+1}^{\log n} |\Gamma^\ell(s)|.
    \]

    \noindent
    With $|\Gamma^\ell(s)|=\binom{\log n}{\ell}$ for all $0\leq\ell \leq \log n$,
    \begin{align*}
         \sum_{\{i,j\} \in E} \left| f_s(i,j) \right|
        =  \sum_{d=0}^{\log n-1} \sum_{\ell=d+1}^{\log n} \binom{\log n}{\ell}
        =  \sum_{\ell=1}^{\log n} \ell\,\binom{\log n}{\ell}
        =  \log n \, 2^{\log n-1}
        =   \log n \, (n/2).
    \end{align*}
   Moreover, it is a well-known result that on hypercubes,
   $\max_{u,v \in V} \rH(u,v) = \Oh(n)$~\cite[p.~372]{Al83}.
   Plugging this into \thmref{flowcovernonregular} yields the claim.
\end{proof}

\newcommand{\textuppertree}{
    For \kary trees ($k\geq2$ constant), $\dEC(G)=\Oh(n \log n)$.
}
\mythm{upper:tree}{\textuppertree}
\begin{proof}
    We examine a complete \kary tree ($k\geq2$)
    of depth $\log_k n-1\in\N$ (the root has depth $0$)
    where the number of nodes is $\sum_{i=0}^{\log_k n-1} k^{i} = n-1$.
    To apply \thmref{flowcovernonregular}, we observe that on a cycle-free graph
    an $\ell_2$-minimal flow $f$ is routed via shortest paths.
    Let us first assume that the distinguished node $s$ of \defref{flow}
    is the root and bound the corresponding optimal flow $f_1$.
    In this case, $f_1(x,y)=k^{\log_k n-i-1}-1$ for $x\in\Gamma^i(s)$ and $y\in\Gamma^{i+1}(s)$.
    Hence,
    \[
        \sum_{ \{i,j\} \in E} | f_1(i,j) |
        = \sum_{d=0}^{\log_k n-1} k^{d} (k^{\log_k n-d-1}-1)
        = \sum_{d=0}^{\log_k n-1} \left(\frac{n}{k} - k^d \right)
        \leq \frac{n}{k} \log_k n.
    \]
    Consider now the more general case, where the distinguished
    vertex $s$ is an arbitrary vertex.
    Here the optimal flow $f$ can be described as a superposition of a flow $f_1$
    and $f_2$, where $f_1$ sends a flow of $n$ tokens from $s$ to the root and
    $f_2$ sends $n-1$ tokens from the root to all other vertices. Clearly,
    $ \sum_{\{i,j\} \in E} | f_2(i,j) | \leq n \log_k n$
    as a flow amount of $n$ is routed over at most $\log_k n$ vertices.
    Therefore,
    \[
      \sum_{ \{i,j\} \in E} | f(i,j) |
      = \sum_{ \{i,j\} \in E} | f_1(i,j) + f_2(i,j) |
      \leq \sum_{ \{i,j\} \in E} | f_1(i,j) | + \sum_{\{i,j\} \in E} | f_2(i,j) |
      \leq \frac{k+1}{k}n \log_k n
    \]
    Moreover, we know from \cite[Proof of Corollary~9]{Zu92} that $\max_{u,v} \rH(u,v) = \Oh(n \log_k n)$. Hence applying \thmref{flowcovernonregular} yields the claim.
\end{proof}

\newcommand{\textupperlollipop}{
    For lollipop graphs, $\dEC(G)=\Oh(n^3)$.
}
\mythm{upper:lollipop}{\textupperlollipop}
\begin{proof}
We use the following strengthened version of \thmref{flowcovernonregular} (see last line of the proof of \thmref{flowcovernonregular}),
\begin{align}
 \dEC_{v} (G) &\leq \max_{u \in V} \rH'(u,v) + \Delta n + \Delta \sum_{ \{i,j\} \in E} | f_v(i,j) |,\label{eq:strength}
\end{align}
where $\dEC_{v}(G)$ refers to a random walk that starts at the vertex $v$. Note that to apply \eq{strength}, we have to consider a random walk with transition matrix $\bP' = \mathbf{I} - \frac{1}{n+1} \, \mathbf{L}$ and hitting times $\rH'(\cdot,\cdot)$.

We first argue why it is sufficient to consider the case where the deterministic walk starts at vertex $v=n/2$. First, if the deterministic walk starts at any other vertex in the complete graph, we know from our upper bound on the deterministic cover time on complete graphs (\thmref{upper:complete}) that after $\Oh(n^2)$ steps, the vertex $n/2$ is reached. Similarly, we know from \thmref{upper:torus} that if the random walk starts at any point of the path, it reaches the vertex $n/2$ within $\Oh(n^3)$ steps (note the extra factor of $\Oh(n)$, as in the corresponding deterministic walk model to~$\bP'$, each node on the path has $n/2+1$ loops).

So let us consider a random walk that starts at vertex $v=n/2$. To apply \eq{strength}, we have to bound $\sum_{\{i,j\} \in E} |f_v(i,j)|$ for a $\ell_2$-optimal flow that sends a flow amount of one from vertex $n/2$ to all other vertices (cf.~\defref{flow}).

Clearly, the $\ell_2$-optimal flow sends at each edge $\{i-1,i\} \in E$, $n/2 < i < n$ in the path a flow of $n-i$. Moreover, it assigns to each edge $(i,n/2)$ with $1 \leq i \leq n/2-1$ a flow of $1$. Hence,
\begin{align*}
 \sum_{\{i,j\} \in E} | f_v(i,j) | &= (n/2-1) \cdot 1 + \sum_{i=n/2+1}^{n} i = \Oh(n^2).
\end{align*}

Our final step is to prove $\max_{u,v \in V} \rH'(u,v)=\Oh(n^3)$ for the symmetric random walk. In fact, we shall prove that this holds for arbitrary graphs. Note that by the symmetry of the transition matrix, $\rH'(u,u)=1/(\pi_u)=n$. So take a shortest path $\mathcal{P}=(u_1=u,u_2,\ldots,u_{\ell}=v)$ of length $\ell$ between $u$ and $v$ in $G$. Note that each time the walk is at any vertex $u_i$ it moves to the vertex $u_{i+1}$ with probability $1/(\Delta+1)$. Hence if $\tau'(u,v)$ describes the random variable for the first hit of $v$ when starting from $u$, we have for any $1 \leq j \leq \ell-1$,
\begin{align*}
  \tau'(u_j,u_{j+1}) &= 1 + \sum_{i=1}^{\mathsf{Geo}(1/(\Delta+1))-1} X_i,
\end{align*}
where $X_i$ is the intermediate time between the $i$-th and $(i+1)$-th visit of $u_j$. Since all $X_i$ are independent and identically distributed random variables with expectation $n$, we can apply Walds equation~\cite{Wald44}
to get
\begin{align*}
 \rH'(u_j,u_{j+1}) &= \Ex{ \tau'(u_i,u_{i+1})} = 1 + \left(\Ex{\mathsf{Geo} \left(\frac{1}{\Delta+1} \right)} - 1 \right) \cdot \Ex{X_i} = 1 + \Delta \cdot n.
\end{align*}
Now using the triangle inequality, we finally get
\begin{align*}
  \rH'(u_1,u_{\ell}) &\leq \sum_{j=1}^{\ell-1} \rH'(u_j,u_{j+1}) \leq (\ell-1) \cdot (1 + \Delta \cdot n) = \Oh(n^3).
\end{align*}

\noindent
Plugging in our findings in \eq{strength}, the claim follows.
\end{proof}

The last theorem about the lollipop graph (a graph that consists of a clique with $n/2$ vertices connected to a path of length $n/2$) might look weak, but turns out to be tight as we will show in \thmref{lower:lollipop}.

\section{Lower Bounds on the Deterministic Cover Time}
\label{sec:lower}

\noindent
We first prove a general lower bound of $\Omega(m)$ on the deterministic cover time
for all graphs.
Afterwards, for all graphs examined in \secref{upper:cover} for which this general
bound is not tight (cycle, path, tree, torus, hypercube, expander)
we present stronger lower bounds which match their respective upper bounds.

\newcommand{\textlowerdense}{
    For all graphs, $\dVC(G)\geq m - \delta$.
}
\mythm{lower:dense}{\textlowerdense}
\begin{proof}
    Let $w$ be a vertex in $G$ with minimum degree $\delta$. Consider the graph $G
    \setminus \{w\}$ with each undirected edge $\{u,v\}$ replaced by a two directed
    edges $(u,v)$ and $(v,u)$. Then there is an Euler tour through $G \setminus
    \{w\}$. We now choose the rotor sequence
    $(\ts(u,1),\ts(u,2),\ldots,\ts(u,\deg(u)))$ of a vertex $u\in V\setminus\{w\}$
    according to the order the neighbors of $u$ are visited by the Euler tour.
    Then the deterministic walk takes the whole
    Euler tour through $G \setminus \{w\}$ of length $m-\deg(w) = m - \delta$
    before visiting $w$.
\end{proof}

As a telling example for a lower bound of the deterministic cover time
of a simple graph, let us
examine a rooted complete \kary tree ($k$ constant).
We choose the rotors to move clockwise and let the walk start at the root.
It is then easy to observe that a configuration where all rotors initially
point downwards towards their respective rightmost successor leads to a order of
explored vertices corresponding to a depth-first-search.
By definition of $\dVC(G)$ this only implies a trivial lower bound
for the deterministic cover time of $\Omega(n)$.
Analogously,
a configuration where all rotors initially point towards the root
leads to a order of the explored vertices corresponding to a breath-first-search.
However, an easy calculation also just gives a linear bound for this walk.

Now consider the following initial configuration: each vertex in the leftmost subtree of
the root is pointing upwards, each vertex in the other subtrees downwards and the
root vertex is pointing to the rightmost subtree.
Then every time the deterministic walk enters one of the subtrees
where the rotors are pointing downwards, it does a depth-first-search walk
of length~$\Theta(n)$.  When it reaches the root again, all inner vertices
of the subtree are visited~$k$ times and all rotors are pointing downwards
again.  On the other hand, when the deterministic walk enters the leftmost subtree
where all rotors are pointing upwards, it only visits one more level than it did
in the previous visit corresponding to a breath-first-search.  Overall,
the leftmost tree is visited $\log_k n$ times and all other vertices are
visited between two visits of the leftmost tree.  This gives a tight
lower bound of $\Omega(n \log n)$ and the following theorem.

\begin{thm}
\label{thm:lower:tree}
    For \kary trees ($k\geq2$ constant), $\dVC(G)=\Omega(n \log n)$.
\end{thm}

A similar analysis gives the following
asymptotically tight lower bounds.

\newcommand{\textlowercycle}{
    For cycles, $\dVC(G)=\Omega(n^2)$.
}
\mythm{lower:cycle}{\textlowercycle}
\begin{proof}
    Let the $n+1$ vertices of an odd cycle be numbered consecutively from $-n/2$ to~$n/2$.
    Consider the initial configuration where every rotor is pointing towards
    the vertex's neighbor with a smaller number in absolute value and the rotor of $0$ point towards $1$.
    Assume that the walk starts from vertex $0$.
    It is easy to see that the sequence of visited vertices by the walk
    consists of $n/2$ phases where
    phase $i$ with $1\leq i <n/2$
    is of length $4i$ and visits $0,1,2,\ldots,i-1,i,i-1,\ldots,2,1,0,-1,-2,\ldots,-(i-1),-i,-(i-1),\ldots,-3,-2,-1$
    while the last phase visits $0,1,2,\ldots,n/2-1,n/2,n/2-1,\ldots,2,1,0,-1,-2,\ldots,-(n/2-1),-n/2$.
    Thus $(n^2+n)/2$ steps are required to cover all vertices.
    Note that the same argument gives a lower bound of $(n-1)^2+1$ steps
    for the path.
\end{proof}

\newcommand{\textlowerlollipop}{
    For lollipop graphs, $\dVC(G)=\Omega(n^3)$.
}
\mythm{lower:lollipop}{\textlowerlollipop}
\begin{proof}
    We number the vertices in the
    clique consecutively from $1$ to~$n/2$, and the vertices on the path
    consecutively from $n/2+1$ to~$n/2$. Further assume that the vertices $n/2$ and
    $n/2+1$ are connected.
    Consider the following configuration: each vertex on the path is pointing
    towards the vertex with a smaller number, and the rotor's permutation of the
    vertices in the complete graph are chosen such that a walk starting from the
    complete graph takes an Eulerian tour therein before escaping to the path. We know from
    the proof of \thmref{lower:cycle} that the root vertex is visited $n$ times,
    before the walk reaches the endpoint $n$. Since everytime the walk returns to
    the complete graph, it takes a complete Euler tour of length $\Theta(n^2)$ there,
    the theorem follows.
\end{proof}

More involved techniques are necessary for expanders, tori and hypercubes.

\newcommand{\textlowerexpander}{
    There are expander graphs with
    $\dVC(G)=\Omega( n \log n)$.
}
\mythm{lower:expander}{\textlowerexpander}

To prove \thmref{lower:expander} we first state the following property
of deterministic walks of \citet{Priezzhev1996}.

\begin{lem}[{\citet[p.~5080]{Priezzhev1996}}]
    \label{lem:Priezzhev}
    Between two successive visits of the same directed edge the unweighted deterministic walk
    visits no other directed edge twice.
\end{lem}
\begin{proof}
    Let the deterministic walk visit $\tx_0, \tx_1, \ldots, \tx_t, \tx_{t+1}$.
    We assume that the last edge $(\tx_t,\tx_{t+1})$ is equal to
    the first edge $(\tx_0,\tx_1)$ and this edge is not visited in between.
    Seeking a contradiction, we assume that there is an edge $(u,v)$
    which is the first edge that is visited twice in between, that is,
    there are (minimal) times $i,j$ with $0<i<j<t$ such that
    $(u,v)=(\tx_i,\tx_{i+1})=(\tx_j,\tx_{j+1})$.
    If the rotor of $u$ pointed twice towards $v$, the walk must have left $u$
    $\deg(u)+1$ often.  Hence the walk must also have entered $u$ that often.
    As there are only $\deg(u)$ edges going from any vertex to~$u$, one of
    these edges must have visited twice, too.  This contradicts our
    assumption that $(u,v)$ was the first edge visited twice.
\end{proof}

\begin{proof}[Proof of \thmref{lower:expander}]
    \newcommand{\ex}{\ensuremath{_{\text{ex}}}}
    \newcommand{\tree}{\ensuremath{_{\text{tr}}}}
    \newcommand{\matching}{\ensuremath{_{\text{ma}}}}
	We construct an expander $G=(V,E)$ with expansion constant $1/20$ and
	prove $\dVC(G)=\Omega( n \log n)$.
	$G$~consists of two subgraphs $G\ex$ and $G\tree$.
	$G\ex=(V\ex,E\ex)$ is a $d$-regular ($d \geq 10$ is a sufficiently large constant)
	expander graph with expansion constant $7/8$ and $n/2$ vertices.
	$G\tree=(V\tree,E\ex)$ is a tree with $n/2$ leaves, where the root has $d$ successors and all other nodes besides leaves have $d-1$ successors.
	Let $E\matching$ be the union of $d$ perfect matchings between the
	leaves of $V\tree$ and $V\ex$.

	Then $V=V\ex\cup V\tree$ and $E=E\ex\cup E\tree \cup E\matching$.
	
	\begin{figure}[!t]
	\begin{center}
	\scalebox{.8}{
	\begin{tikzpicture}[thick,auto,inner sep=0.75mm,domain=0:3,x=1cm,y=1cm]
\pgftransformxscale{0.55}
\pgftransformyscale{0.55}
\draw [rounded corners=15pt, fill=black!20] (4,9.1) -- (2.25,9.1) -- (-1,5.2) -- (9,5.2) -- (5.75,9.1) -- (4,9.1);
\draw [rounded corners=20pt, fill=black!20] (4,4) -- (-1,4) -- (-1,1.5) -- (9,1.5) -- (9,4) -- (4,4);
\node (1) at (4,8) [circle, draw, fill=black] {};
\node (2) at (1.5,7) [circle, draw, fill=black] {};
\node (3) at (4,7) [circle, draw, fill=black] {};
\node (4) at (6.5,7) [circle, draw, fill=black] {};
\node (5) at (1,6) [circle, draw, fill=black] {};
\node (6) at (2,6) [circle, draw, fill=black] {};
\node (7) at (3.5,6) [circle, draw, fill=black] {};
\node (8) at (4.5,6) [circle, draw, fill=black] {};
\node (9) at (6,6) [circle, draw, fill=black] {};
\node (10) at (7,6) [circle, draw, fill=black] {};
\node (15) at (1,3) [circle, draw, fill=black] {};
\node (16) at (2,3) [circle, draw, fill=black] {};
\node (17) at (3.5,3) [circle, draw, fill=black] {};
\node (18) at (4.5,3) [circle, draw, fill=black] {};
\node (19) at (6,3) [circle, draw, fill=black] {};
\node (20) at (7,3) [circle, draw, fill=black] {};
\draw (1) -- (2);
\draw (1) -- (3);
\draw (1) -- (4);
\draw (2) -- (5);
\draw (2) -- (6);
\draw (3) -- (7);
\draw (3) -- (8);
\draw (4) -- (9);
\draw (4) -- (10);
    \draw [dashed] (5) -- (16);
    \draw [dashed] (5) -- (19);
    \draw [dashed] (5) -- (20);
    \draw [dashed] (6) -- (17);
    \draw [dashed] (6) -- (15);
    \draw [dashed] (6) -- (20);
    \draw [dashed] (7) -- (18);
    \draw [dashed] (7) -- (15);
    \draw [dashed] (7) -- (17);
    \draw [dashed] (8) -- (19);
    \draw [dashed] (8) -- (20);
    \draw [dashed] (8) -- (16);
    \draw [dashed] (9) -- (18);
    \draw [dashed] (9) -- (19);
    \draw [dashed] (9) -- (17);
    \draw [dashed] (10) -- (15);
    \draw [dashed] (10) -- (16);
    \draw [dashed] (10) -- (18);

    \draw  (20) arc(-30:-150:3.5cm and 2.5cm);
    \draw  (20) arc(-30:-150:2cm and 1.375cm);
    \draw  (20) arc(-30:-150:1.5cm and 0.75cm);
    \draw  (19) arc(-30:-150:2.9cm and 1.5cm);
    \draw  (19) arc(-30:-150:1.5cm and 0.7cm);
    \draw  (19) arc(-30:-150:0.875cm and 0.5cm);
    \draw  (18) arc(-30:-150:1.5cm and 1cm);
    \draw  (17) arc(-30:-150:0.875cm and 0.5cm);
    \draw  (16) arc(-30:-150:0.55cm and 0.25cm);
    \draw (4.5,8.5) node[right] {$G\tree$};
    \draw (6.5,4.5) node[right] {$E\matching$};
    \draw (7.5,2.9) node[right] {$G\ex$};
    \end{tikzpicture}}
    \end{center}
    \caption{An illustration of the expander graph $G$ used in \thmref{lower:expander} for $d=3$ and $n=12$.}
\end{figure}
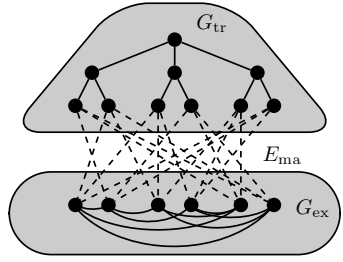

	We first prove that such a graph exists, prove some properties and
	that $G$ is an expander itself. At the end we prove the bound for $\dVC(G)$.

    We choose $G\ex$ as a $d$-regular Ramanujan graph with expansion constant at least $7/8$ and $n/2$ vertices
    i.e.,
    $
    |E(X,X^c)| \geq \frac{7}{8} d \,|X|,
	$
	for all $X \subseteq V\ex$ with $1 \leq |X| \leq n/4$.
	Such a graph exists since for random $d$-regular graphs, $\lambda_2=\O(d^{-1/2})$~\citep{DS91}
	and moreover
	$
	    |E(X,X^c)|/(d\,|X|)
	    \geq \sqrt{1-\lambda_2}
	    \geq \sqrt{1-\O(d^{-1/2})}$.
	
	Let us now consider a set $X \subseteq V\ex$ with $(1/4) n \leq |X| \leq (3/8) n$. Then,
	\begin{align}
	 |E(X,X^c)| &\geq \frac{7}{8} d \frac{1}{4} n - d \left( |X| - \frac{1}{4} n  \right )
	   \geq \frac{7}{32} d n - \frac{1}{8} dn
	   = \frac{3}{32} dn
	   \geq \frac{3}{32} d \frac{8}{3}\,|X|
	   = \frac{1}{4} d\,|X|.
	  \label{eq:strongexpander}
	\end{align}	
	\noindent
	To calculate $|V|$, observe that the total number of vertices in $G\tree$ is
	\[
	    |V\tree|
	    \leq \sum_{i=0}^{ \log_{d}(n/2)} d^i
	    = \frac{d(n/2)-1}{d-1}
	    \leq \frac{d}{d-1} \, \frac{n}{2}
	  	\leq \frac{9}{8} \, \frac{n}{2}
	  	= \frac{9}{16} n
	\]
	and therefore $G$ has $|V|\leq(17/16) n \leq (9/8) n$ vertices with
	$\delta(G)=d$ and $\Delta(G)=2d$.
	To see that $G$ is also an expander graph,
	take any subset $X \subseteq V$ with $1 \leq |X| \leq (9/16) n$.
	Let $X\tree := X \cap V\tree$ and $X\ex := X \cap V\ex$.
	\begin{enumerate}
	 \item Consider first the case where $|X\tree| \geq 4 |X\ex|$.
	 Observe that for any $X\tree$, $|E(X\tree,X\tree)| \leq |X\tree|-1$ as $G\tree$ is a tree.
	 Therefore,
	 \begin{align*}
	   |E(X,X^c)|
	   &\geq |E(X\tree,X^c)|
	   =|E(X\tree,V)| - |E(X\tree,X\tree)| - |E(X\tree,X\ex)|\\
	   &\geq d \, |X\tree| - |X\tree| - d \, |X\ex|
       \geq d |X\tree| - |X\tree| - \frac{d}{4} \, |X\tree| \\
	   &\geq \left(\frac{3}{4} d - 1 \right) \, |X\tree|
	   \geq \left(\frac{3}{4} d - 1 \right) \frac{4}{5} \, |X|
	   \geq \frac{13}{25} d \, |X|.
	 \end{align*}
	\item Assume now that $|X\tree| \leq 4 \, |X\ex|$. If $|X\ex| \leq (3/8) n$, then \eq{strongexpander} implies that
	\begin{align*}
	 |E(X,X^c)|
	 \geq  |E(X\ex,V\ex\setminus X)|
	 =  |E(X\ex,V\ex\setminus X\ex)|
	 \geq \frac{1}{4} d \,|X\ex|
	 \geq \frac{1}{20} d \,|X|.
	\end{align*}
	On the other hand, if $X\ex \geq (3/8) n$, it follows that $X\tree \leq (3/16) n$. Since each vertex in $X\ex$ has $d$ edges to~$V\tree$, we have
	\begin{align*}
	  |E(X,X^c)|
	  &\geq |E(X\ex,X^c)|
	  \geq |E(X\ex,X\ex^c)| - |E(X\ex,X\tree)|
	  \geq \frac{3}{8} d n  - \frac{3}{16} (d+1) n\\	
	  &= \frac{3}{16}\,(d-1)\,n
	  \geq \frac{3}{16}\,(d-1)\,\frac{16}{9}|X|
	  = \frac{1}{3}\,(d-1)\,|X|
	  \geq \frac{3}{10} \, d \, |X|.
	\end{align*}
	\end{enumerate}
	Hence we conclude that the graph $G$ is an expander graph with expansion constant $1/20$.
	
    We are now ready to define the rotors.
    As in the proof of \thmref{lower:dense}, choose an Euler tour of the directed
    graph $G\ex$ and set the rotors of $V\ex$ and the initial position
    such that the deterministic walk on $G$ first performs an Euler tour on $G\ex$
    before visiting any node from $V\tree$.  For vertices from $V\tree$ we choose
    the rotor sequence similar to the proof of \thmref{lower:tree} such that
    the direction of the root is always the last one in the sequence.
    Let $u_i\in V\tree$ be the first node in level~$i$ with
    $0<i<\log_d(n/2)-2$ of the tree $G\tree$ which is reached.
    Let this happen at time $t_i$ from a node $u_{i+1}$ in level $i+1$.
    Let $t_{i+1}$ be the first time $u_{i+1}$ is visited.
    By choice of the rotor sequence, only at the $(d+1)$-th visit
    to~$u_{i+1}$ its rotor can point upwards to~$u_i$.  As $u_{i+1}$
    has only $d$ children, one child $u_{i+2}$ must be visited twice
    between times $t_{i+1}$ and $t_i$.  Hence also the directed
    edge $(u_{i+2},u_{i+1})$ is visited twice in this time interval.
    Assume there was an edge $e\in E\ex$ which was not visited
    in this time interval.  We know that this edge $e$ is visited before
    time $t_{i+1}$ by the Euler tour and that it is visited after time
    $t_i$ as the graph is strongly connected and the deterministic walk
    eventually visits all edges arbitrarily often.  \lemref{Priezzhev}
    implies that then $e$ must also be visited
    between times $t_{i+1}$ and $t_i$.  Overall, between every new level of
    $V\tree$ which is explored, the deterministic walk has to visit
    all edges $E\ex$.  Hence it takes $\Omega(n\log n)$ steps to visits
    all vertices of $G$.
\end{proof}

\newcommand{\textlowertorus}{
    For two-dimensional torus graphs, $\dVC(G)=\Omega(n^{3/2})$.
}
\mythm{lower:torus}{\textlowertorus}
\begin{proof}
    \newcommand{\Cout}{\textsf{out}\xspace}
    \newcommand{\Cin}{\textsf{in}\xspace}
    \newcommand{\Ccycle}{\textsf{cycle}\xspace}

    \newcommand{\Fstyle}[1]{{\underline{#1}}\xspace}
    \newcommand{\Fout}{\Fstyle{\Cout}\xspace}
    \newcommand{\Fin}{\Fstyle{\Cin}\xspace}
    \newcommand{\Fcycle}{\Fstyle{\Ccycle}\xspace}

    \newcommand{\mydown}{\begin{tikzpicture}[scale=.5,baseline=0pt]
        \filldraw[color=white] (0.0,0.0) rectangle (0.57,0.57);
        \filldraw[color=black] (0.25,0.5) circle (.07cm);
        \draw[->,thick] (0.25,0.5) -- (0.25,0);
        \end{tikzpicture}\xspace}
    \newcommand{\myleft}{\begin{tikzpicture}[scale=.5,baseline=0pt]
        \filldraw[color=white] (0,0) rectangle (0.57,0.57);
        \filldraw[color=black] (0.5,0.25) circle (.07cm);
        \draw[->,thick] (0.5,0.25) -- (0,0.25);
        \end{tikzpicture}\xspace}
    \newcommand{\myup}{\begin{tikzpicture}[scale=.5,baseline=0pt]
        \filldraw[color=white] (-0.07,-0.07) rectangle (0.5,0.5);
        \filldraw[color=black] (0.25,0) circle (.07cm);
        \draw[->,thick] (0.25,0) -- (0.25,0.5);
        \end{tikzpicture}\xspace}
    \newcommand{\myright}{\begin{tikzpicture}[scale=.5,baseline=0pt]
        \filldraw[color=white] (-0.07,-0.07) rectangle (0.5,0.5);
        \filldraw[color=black] (0,0.25) circle (.07cm);
        \draw[->,thick] (0,0.25) -- (0.5,0.25);
        \end{tikzpicture}\xspace}

    Consider a two-dimensional $\sqrt{n} \times \sqrt{n}$ torus.
    For simplicity we assume that $\sqrt{n}$ is an odd integer and
    represent the vertices by two coordinates $(x,y)$ with $-L \leq
    x,y \leq L$ with $L:=(\sqrt{n}-1)/2$.

    \newcommand{\myspace}{\,\,}  
    \begin{table}[tb]
    \begin{center}
    \footnotesize
    \begin{tabular}{|cc|cccccc|c|}
    \hline\hline
    \multirow{2}{*}{\bf Phase} &
    \multirow{2}{*}{\bf Steps} &
    \multicolumn{6}{|c|}{\bf Situation at the end of the respective phase} &
    \multirow{2}{*}{\bf Corresponding Figure}\\
    &&
    \myspace $\bm{C_1}$ \myspace\myspace &
    \myspace $\bm{C_2}$ \myspace\myspace &
    \myspace $\bm{C_3}$ \myspace\myspace &
    \myspace $\bm{C_4}$ \myspace\myspace &
    \myspace $\bm{C_5}$ \myspace\myspace &
    \myspace $\bm{C_6}$ \myspace\myspace &
    \\
    \hline\hline
     1 &     1    & \Fin    & \Cin    & \Cin & \Cin & \Cin & \Cin & \figref{torus:1} \\
     2 &   2--9   & \Fcycle & \Cin    & \Cin & \Cin & \Cin & \Cin & \figref{torus:9} \\
     3 &  10--18  & \Cout   & \Fin    & \Cin & \Cin & \Cin & \Cin  & \figref{torus:18} \\
     4 &  19--49  & \Fin    & \Ccycle & \Cin & \Cin & \Cin & \Cin  & \figref{torus:49} \\
     5 &  50--57  & \Fcycle & \Ccycle & \Cin & \Cin & \Cin & \Cin  & \figref{torus:57} \\
     6 &  58--66  & \Cout   & \Fcycle & \Cin & \Cin & \Cin & \Cin  & \figref{torus:66} \\
     7 &  67--83  & \Cout   & \Cout   & \Fin & \Cin & \Cin & \Cin  & \figref{torus:83} \\
     8 &  84--138 & \Cout   & \Fin    & \Ccycle & \Cin & \Cin & \Cin  & \figref{torus:138} \\
     9 & 139--169 & \Fin    & \Ccycle & \Ccycle & \Cin & \Cin & \Cin  & \figref{torus:169} \\
    10 & 170--177 & \Fcycle & \Ccycle & \Ccycle & \Cin & \Cin & \Cin  & \figref{torus:177} \\
    11 & 178--186 & \Cout   & \Fcycle & \Ccycle & \Cin & \Cin & \Cin  & \figref{torus:186} \\
    12 & 187--203 & \Cout   & \Cout   & \Fcycle & \Cin & \Cin & \Cin  & \figref{torus:203} \\
    \hline\hline
    \end{tabular}
    \vspace*{0.5\baselineskip}
    \caption{First twelve phases of a deterministic walk on the two-dimensional torus.
    A visual description of the phases on the $7\times 7$ torus is given in \figref{lower:torus}.
    The different states are defined in the proof of \thmref{lower:torus}.
    Underlined states indicate the last position of the deterministic walk.
    }
    \vspace*{-1.5\baselineskip}
    \label{tab:lower:torus}
    \end{center}
    \end{table}

    Let all rotor sequences be ordered clockwise (that is, \myup,\myright,\mydown,\myleft,\ldots)
    and start with a rotor in the direction of the origin $(0,0)$.
    More precisely, let the initial rotor direction
    at vertex $(x,y)$ be

    \noindent
    \begin{tabular}{l@{\qquad\qquad}l}
    \begin{minipage}[c]{0.6944\linewidth}
    \begin{enumerate}
        \setlength{\itemsep}{0pt}
        \setlength{\parskip}{0pt}
        \item \myup if $y \leq -1$ and $y \leq -x$ and $y < x$ or $(x,y)=(0,0)$,
        \item \myright if $x \leq -1$ and $-x > y$ and $x \leq y$,
        \item \mydown if $y \geq 1$ and $y \geq -x$ and $y > x$,
        \item \myleft if $x \geq 1$ and $-x < y$ and $x \geq y$.
    \end{enumerate}
    \end{minipage}
    &
    \input{lowertorus_initial.tex}
    \\
    \end{tabular}

    \begin{figure}[!p]
        \input{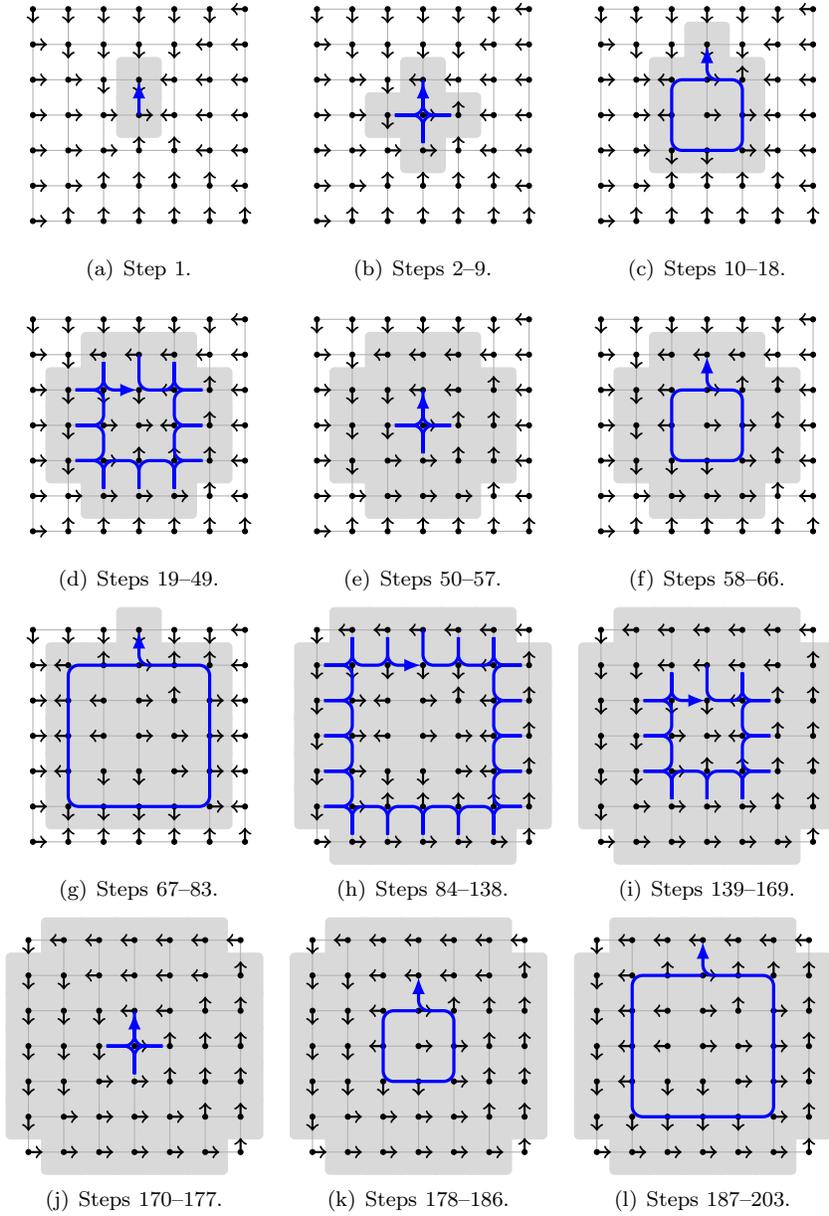}

        \caption{The first eleven phases of the deterministic walk
            on the two-dimensional $7 \times 7$ torus.  All rotors are initially
            pointing towards the origin.
            In each phase the deterministic walk is shown as a blue arrow.
            The depicted rotors correspond to the rotor directions at the end
            of the respective phase.
            The gray shaded area marks all covered vertices at this time.
            }
        \label{fig:lower:torus}
    \end{figure}

    \noindent
    We will start the random walk from the origin $(0,0)$.
    Denote by $C_i:=\{ (x,y)\in V \colon \max\{x,y\}=i\}$
    the boundary of the square defined by the corners
    $(i,i)$, $(i,-i)$,$(-i,-i)$,$(-i,i)$.
    For each square $C_i$, $1 \leq i \leq L$,
    we define three different states called \Cin, \Ccycle and \Cout:

    \noindent
    \begin{tabular}{r@{\ \:}l@{\ \ }l@{\qquad}l}
    (i) & $C_i$ is \Cin, iff &
    $\rho(x,y) =
    \begin{cases}
     \mydown & \mbox{for $y=i$ and $-i \leq x < i$}, \\
     \myleft & \mbox{for $x=i$ and $-i < y \leq i$}, \\
     \myup & \mbox{for $y=-i$ and $-i < x \leq i$}, \\
     \myright & \mbox {for $x=-i$ and $-i \leq y < i$ }. \\
    \end{cases}$ &
    \input{lowertorus_in.tex}
    \\[1cm]
    (ii) & $C_i$ is \Ccycle, iff &
    $\rho(x,y) =
    \begin{cases}
     \mydown & \mbox{for $x=-i$ and $-i < y \leq i$}, \\
     \myleft & \mbox{for $y=i$ and $-i < x \leq i$}, \\
     \myup & \mbox{for $x=i$ and $-i \leq y < i$}, \\
     \myright & \mbox {for $y=-i$ and $-i \leq x < i$ }. \\
    \end{cases}$ &
    \input{lowertorus_cycle.tex}
    \\[1cm]
    (iii) & $C_i$ is \Cout, iff &
    $\rho(x,y) =
    \begin{cases}
     \myleft & \mbox{for $x=-i$ and $-i < y \leq i$}, \\
     \myup & \mbox{for $y=i$ and $-i < x \leq i$ and $x\neq0$}, \\
     \multirow{2}{*}{$\myright$} & \mbox{for $x=i$ and $-i \leq y < i$} \\
              & \mbox{\phantom{for }or $y=i$ and $x=0$}, \\
     \mydown & \mbox {for $y=-i$ and $-i \leq x < i$ }. \\
    \end{cases}$
    &
    \input{lowertorus_out.tex}
    \end{tabular}
    \vspace{.5\baselineskip}

    \noindent By definition, the initial rotor directions of all squares $C_i$,
    $1 \leq i \leq L$,  is \Cin.
    We decompose the deterministic walk in phases such that after each phase
    the walk is at a vertex $(0,y)$ for some $y$ with $1\leq y\leq L$,
    every square $C_i$, $1 \leq i \leq L$, has a well-defined
    state (\Cin, \Ccycle or \Cout), and the rotor at $(0,0)$ is \myright.
    The first phase has length one.
    Hence after the first phase the rotors are
    $(C_1,C_2,C_3,\ldots)=(\Fin,\Cin,\Cin,\ldots)$ where
    the underline in $\Fin$ marks the current position of the walk.
    We now observe the following three simple rules which can be easily proven by induction.
    \begin{enumerate}
        \setlength{\itemsep}{0pt}
        \setlength{\parskip}{0pt}
        \item If the walk is at $(0,1)$ and $C_1=\Cin$,
              then after eight steps the walk is at $(0,2)$ and $C_1=\Ccycle$.
              Or in short: $(\Fin,\ldots) \xLongrightarrow{8} (\Fcycle,\ldots)$.     \item If the walk is at $(0,y)$, $y\geq1$, and $C_y=\Ccycle$,
              then after $8y+1$ steps the walk is at $(0,y+1)$ and $C_y=\Cout$.
              Or in short: $(\ldots,\Fcycle,C_{y+1},\ldots)
              \xLongrightarrow{8y+1} (\ldots,\Cout,\Fstyle{C_{y+1}},\ldots)$.
        \item If the walk is at $(0,y)$, $y\geq2$, and $C_y=\Cin$ as well as $C_{y-1}=\Cout$,
              then after $24y-17$ steps the walk is at $(0,y-1)$ and $C_y=\Ccycle$ as well as $C_{y-1}=\Cin$.
              Or in short: $(\ldots,C_{y-1},\Fin,\ldots)
              \xLongrightarrow{24y-17} (\ldots,\Fin,\Ccycle,\ldots)$.
    \end{enumerate}

    \noindent \tabref{lower:torus}
    shows the first twelve phases applying above three rules.
    In the introduced short notation, after the first phase the rotors are
    $(\Fin,\Cin^{L-1})$.  The second phase applies rule (i) and reaches
    $(\Fin,\Cin^{L-1}) \xLongrightarrow{8} (\Fcycle,\Cin^{L-1})$.
    The three subsequent phases can be described as follows
    (corresponding to \figref{lower:torus} \subref{fig:torus:9}--\subref{fig:torus:57}):
    \[
        (\Fcycle,\Cin^{L-1})
        \xLongrightarrow{9}
        (\Cout,\Fin,\Cin^{L-2})
        \xLongrightarrow{31}
        (\Fin,\Ccycle,\Cin^{L-2})
        \xLongrightarrow{8}
        (\Fcycle,\Ccycle,\Cin^{L-2}).
    \]
    These three phases (or 48 steps) already reveal the general pattern.
    By induction one can prove for all $k$ with $1\leq k<L$:
    \begin{align*}
        (\Fcycle,\Ccycle^{k-1}, \Cin^{L-k})
        &\xLongrightarrow{\sum_{y=1}^k (8y+1)}
        (\Cout^k, \Fin, \Cin^{L-k-1})
        \xLongrightarrow{\sum_{y=2}^{k+1} (24y-17)}
        (\Fin,\Ccycle^{k}, \Cin^{L-k-1})\\
        &\xLongrightarrow{8}
        (\Fcycle,\Ccycle^{k}, \Cin^{L-k-1}).
    \end{align*}
    This shows that the deterministic walk needs
    $\sum_{y=1}^k (8y+1) + \sum_{y=2}^{k+1} (24y-17) + 8
    = 8\,(2k+1)\,(k+1)$ steps to go from
    $(\Fcycle,\Ccycle^{k-1}, \Cin^{L-k})$
    to
    $(\Fcycle,\Ccycle^{k}, \Cin^{L-k-1})$.
    To get a lower bound on the deterministic cover time,
    we bound the time to reach $(0,L)$ with $C_L=\Ccycle$:
    \begin{align*}
        (\Fin,\Cin^{L-1})
        &\xLongrightarrow{8}
        (\Fcycle,\Cin^{L-1})
        \xLongrightarrow{\sum_{k=1}^{L-1} 8\,(2k+1)\,(k+1)}
        (\Fcycle,\Ccycle^{L-1})\\
        &\xLongrightarrow{\sum_{y=1}^{L-1} 8y+1}
        (\Ccycle^{L-1},\Fcycle).
    \end{align*}
    After $(\Ccycle^{L-1},\Fcycle)$ is reached, the deterministic walk
    only needs $7L$ further steps to go from $(0,L)$ along $C_L=\Ccycle$
    to the last uncovered vertex $(L,L)$.  This gives an overall lower
    bound for the deterministic cover time of
    \[
        1 +
        8 +
        \tsum_{k=1}^{L-1} 8\,(2k+1)\,(k+1) +
        \tsum_{y=1}^{L-1} 8y+1 +
        7L
        =
        \tfrac{16}{3} L^3 + 8 L^2 + \tfrac{8}{3}L
        =
        \tfrac{2}{3} \, (n^{3/2} - \sqrt{n}).
        \qedhere
    \]
\end{proof}

\newcommand{\textlowerhypercube}{
    For hypercubes, $\dVC(G)=\Omega(n \log^2 n)$.
}
\mythm{lower:hypercube}{\textlowerhypercube}
\begin{proof}
    We consider the $d$-dimensional hypercube with $n=2^d$
    vertices corresponding to bitstrings $\{0,1\}^d$. A pair of vertices is connected
    if their bitstrings differ in exactly one bit.

\begin{figure}[p]
    \rotatebox{90}{\scalebox{.88}{
    \begin{tikzpicture}[thick,level distance=13mm,
        level 1/.style={sibling distance=40mm},
        level 2/.style={sibling distance=10.5mm},
        level 3/.style={sibling distance=3.5mm},
        level 4/.style={sibling distance=1.8mm},
        grow=down]
        \tikzstyle{every node}=[font=\tiny]
        \node{00000}
        child { node {00001}
        child { node[rotate=270] {00011}
        child { node[rotate=270] {00111}
        child { node[rotate=270] {01111}
        child { node[rotate=270] {11111} } }
        child { node[rotate=270] {10111}
        child { node[rotate=270] {11111} } } }
        child { node[rotate=270] {01011}
        child { node[rotate=270] {01111} }
        child { node[rotate=270] {11011}
        child { node[rotate=270] {11111} } } }
        child { node[rotate=270] {10011}
        child { node[rotate=270] {10111} }
        child { node[rotate=270] {11011} } } }
        child { node[rotate=270] {00101}
        child { node[rotate=270] {00111} }
        child { node[rotate=270] {01101}
        child { node[rotate=270] {01111} }
        child { node[rotate=270] {11101}
        child { node[rotate=270] {11111} } } }
        child { node[rotate=270] {10101}
        child { node[rotate=270] {10111} }
        child { node[rotate=270] {11101} } } }
        child { node[rotate=270] {01001}
        child { node[rotate=270] {01011} }
        child { node[rotate=270] {01101} }
        child { node[rotate=270] {11001}
        child { node[rotate=270] {11011} }
        child { node[rotate=270] {11101} } } }
        child { node[rotate=270] {10001}
        child { node[rotate=270] {10011} }
        child { node[rotate=270] {10101} }
        child { node[rotate=270] {11001} } } }
        child { node {00010}
        child { node[rotate=270] {00011} }
        child { node[rotate=270] {00110}
        child { node[rotate=270] {00111} }
        child { node[rotate=270] {01110}
        child { node[rotate=270] {01111} }
        child { node[rotate=270] {11110}
        child { node[rotate=270] {11111} } } }
        child { node[rotate=270] {10110}
        child { node[rotate=270] {10111} }
        child { node[rotate=270] {11110} } } }
        child { node[rotate=270] {01010}
        child { node[rotate=270] {01011} }
        child { node[rotate=270] {01110} }
        child { node[rotate=270] {11010}
        child { node[rotate=270] {11011} }
        child { node[rotate=270] {11110} } } }
        child { node[rotate=270] {10010}
        child { node[rotate=270] {10011} }
        child { node[rotate=270] {10110} }
        child { node[rotate=270] {11010} } } }
        child { node {00100}
        child { node[rotate=270] {00101} }
        child { node[rotate=270] {00110} }
        child { node[rotate=270] {01100}
        child { node[rotate=270] {01101} }
        child { node[rotate=270] {01110} }
        child { node[rotate=270] {11100}
        child { node[rotate=270] {11101} }
        child { node[rotate=270] {11110} } } }
        child { node[rotate=270] {10100}
        child { node[rotate=270] {10101} }
        child { node[rotate=270] {10110} }
        child { node[rotate=270] {11100} } } }
        child { node {01000}
        child { node[rotate=270] {01001} }
        child { node[rotate=270] {01010} }
        child { node[rotate=270] {01100} }
        child { node[rotate=270] {11000}
        child { node[rotate=270] {11001} }
        child { node[rotate=270] {11010} }
        child { node[rotate=270] {11100} } } }
        child { node {10000}
        child { node[rotate=270] {10001} }
        child { node[rotate=270] {10010} }
        child { node[rotate=270] {10100} }
        child { node[rotate=270] {11000} } } ;
        \draw[thick,loosely dashed] (-10.1,-1.88) --(+9.8,-1.88)
         node[anchor=west,fill=white,text width=1.75cm,text ragged] {depth limit for $1$-st phase};
        \draw[thick,loosely dashed] (-10.1,-3.25) -- (+9.8,-3.25)
         node[anchor=west,fill=white,text width=1.75cm,text ragged] {depth limit for $2$-nd phase};
        \draw[thick,loosely dashed] (-10.1,-4.55) -- (+9.8,-4.55)
         node[anchor=west,fill=white,text width=1.75cm,text ragged] {depth limit for $3$-rd phase};
        \draw[thick,loosely dashed] (-10.1,-5.85) -- (+9.8,-5.85)
         node[anchor=west,fill=white,text width=1.75cm,text ragged] {depth limit for $4$-th phase};
        \end{tikzpicture}
    }}
    \hspace*{1.5cm}
    \rotatebox{90}{\scalebox{.88}{
        \tiny
        \begin{tabular}{|l|p{19.7cm}|}
        \hline
        $1$-st phase &
        00000 00001 00000 00010 00000 00100 00000 01000 00000 10000 \\
        \hline
        $2$-nd phase &
        00000 00001 00011 00001 00101 00001 01001 00001 10001 00001 00000 00010 00011 00010 00110 00010 01010 00010 10010 00010 00000 00100 00101 00100 00110 00100 01100 00100 10100 00100 00000 01000 01001 01000 01010 01000 01100 01000 11000 01000 00000 10000 10001 10000 10010 10000 10100 10000 11000 10000 \\
        \hline
        $3$-rd phase &
        00000 00001 00011 00111 00011 01011 00011 10011 00011 00001 00101 00111 00101 01101 00101 10101 00101 00001 01001 01011 01001 01101 01001 11001 01001 00001 10001 10011 10001 10101 10001 11001 10001 00001 00000 00010 00011 00010 00110 00111 00110 01110 00110 10110 00110 00010 01010 01011 01010 01110 01010 11010 01010 00010 10010 10011 10010 10110 10010 11010 10010 00010 00000 00100 00101 00100 00110 00100 01100 01101 01100 01110 01100 11100 01100 00100 10100 10101 10100 10110 10100 11100 10100 00100 00000 01000 01001 01000 01010 01000 01100 01000 11000 11001 11000 11010 11000 11100 11000 01000 00000 10000 10001 10000 10010 10000 10100 10000 11000 10000 \\
        \hline
        $4$-th phase &
        00000 00001 00011 00111 01111 00111 10111 00111 00011 01011 01111 01011 11011 01011 00011 10011 10111 10011 11011 10011 00011 00001 00101 00111 00101 01101 01111 01101 11101 01101 00101 10101 10111 10101 11101 10101 00101 00001 01001 01011 01001 01101 01001 11001 11011 11001 11101 11001 01001 00001 10001 10011 10001 10101 10001 11001 10001 00001 00000 00010 00011 00010 00110 00111 00110 01110 01111 01110 11110 01110 00110 10110 10111 10110 11110 10110 00110 00010 01010 01011 01010 01110 01010 11010 11011 11010 11110 11010 01010 00010 10010 10011 10010 10110 10010 11010 10010 00010 00000 00100 00101 00100 00110 00100 01100 01101 01100 01110 01100 11100 11101 11100 11110 11100 01100 00100 10100 10101 10100 10110 10100 11100 10100 00100 00000 01000 01001 01000 01010 01000 01100 01000 11000 11001 11000 11010 11000 11100 11000 01000 00000 10000 10001 10000 10010 10000 10100 10000 11000 10000 \\
        \hline
        $5$-th phase &
        00000 00001 00011 00111 01111 11111 done\\
        \hline
        \end{tabular}
    }}

    \caption{Illustration for the proof of \thmref{lower:hypercube} showing a
        deterministic walk on the hypercube~$H_5$.
        The walk starts at $00000$ and all rotor sequences are sorted lexicographically.}
    \label{fig:H5}
\end{figure}

    We first note that if
    all rotors are set up the same way (w.r.t.\ the dimension),
    all vertices of the hypercube are covered within $\O(n)$ steps.
    Hence we choose the rotor sequences differently,
    that is,
    let all rotor sequences be sorted lexicographically and
    let the walk start at $0^d$.  We bound the time to reach~$1^d$.
    The chosen rotor sequence implies that every vertex~$x$
    first visits the neighboring vertices~$y$ with
    $|y|_1<|x|_1$ and then the neighboring vertices~$y$ with $|y|_1>|x|_1$.
    The resulting walk can be nicely described as a sequence of depth first searches (DFS)
    on a tree~$T_d$ which is defined as follows.  The vertices are bitstrings $\{0,1\}^d$
    corresponding to the vertices of~$H_d$.  If there is an edge $\{u,v\}$ in~$T_d$, then
    $\{u,v\}$ is an edge in~$H_d$, too (but not the other way around).
    The root of~$T_d$ is~$0^d$.
    In level~$i$ there are only vertices~$x$ with $|x|_1=i$.
    Every vertex~$x$ has either 0 or~$|x|_0$ children
    corresponding to neighbors~$y$ of~$x$ with $|y|_1=|x|_1+1$.
    The root has~$d$ children.
    All other vertices~$x$ have~$|x|_0$ children iff the single bit in which it differs from its parent
    is left of the leftmost~$1$-bit.  That is, the tree is truncated if
    the bit which is flipped from the parent to the child is not a leading zero.
    Note that this implies that on level~$i$ there are~$i$ copies of each vertex with $|x|_1=i$
    in~$T_d$ and only the leftmost recurses to the next level.
    \figref{H5} on page~\pageref{fig:H5} shows
    the deterministic walk on~$H_5$ and the corresponding tree~$T_5$
    as defined above.

    We decompose the deterministic walk on~$H_d$ in~$d$ phases such that phase~$i$
    has length $2\sum_{j=1}^i j\,\binom{d}{j}$ and show the the following:
    \begin{enumerate}
    \setlength{\itemsep}{0pt}
    \setlength{\parskip}{0pt}
    \item Initially, the rotors corresponding to the vertices in~$T_d$ point towards
          their respective parent.
    \item After the $i$\nobreakdash-th phase, all rotors of vertices~$x$ with $|x|_1> i$
          point towards their respective parent in~$T_d$ while
          all rotors of vertices $x$ with $|x|_1\leq i$ point to their leftmost child in~$T_d$,
          i.e., to their lexicographic smallest neighbor with one more bit set to one.
    \item The $i$\nobreakdash-th phase of the deterministic walk on~$H_d$
          visits the same vertices in the same order as a
          DFS on~$T_d$ with limited depth~$i$.
    \item $T_d$ has $i\,\binom{d}{i}$ vertices on level~$i$ with $i>0$.
    \end{enumerate}

    \noindent
    (i) holds by definition of~$T_d$ for the chosen rotor sequence.
    We now prove simultaneously (ii) and (iii) by induction.
    For the first phase of~$2d$ steps it is easy to see as it
    alternates between~$0^d$ and all nodes~$x$ with $|x|_1=1$ (in increasing order).
    This phase ends at the root~$0^d$ whose rotor then points at $0^{d-1}1$ (as initially).
    The rotors of nodes~$x$ with $|x|_1=1$ now point to the lexicographically
    smallest neighbor~$y$ with $|y|_1=2$.  Note that for every vertex~$y$ with $|y|_1=2$
    there are two vertices~$x$ with $|x|_1=1$ whose rotor points to~$y$.
    Let us now assume (ii) and (iii) holds after the $i$\nobreakdash-th phase.
    Then all rotors of vertices~$x$ with $|x|_1\leq i$ point downwards in~$T_d$
    to their leftmost child which is also the lexicographically smallest.
    It is obvious that then the deterministic walk on~$H_d$ exactly performs a DFS
    of limited depth~$i+1$ on~$H_d$ up to the point when a vertex is visited the
    second time within this phase.  If this vertex~$x$ is in the last visited layer, i.e.,
    $|x|=i+1$, then its rotor points upwards and the DFS is not disturbed.
    If a vertex~$x$ with $|x|\leq i$ is visited the second time, then by definition
    of~$T_d$, this vertex has no children.  Hence the DFS in~$T_d$ goes
    back to its parent which is the same vertex to which the deterministic walk
    in~$H_d$ moves as the rotor is already pointing at its second neighbor in its rotor sequence.
    The same holds for the third, fourth, and so on visit of a vertex.
    Overall, this $(i+1)$\nobreakdash-th phase visits all vertices in~$T_d$ up to depth~$(i+1)$ and
    changes the rotors of vertices~$x$ with $|x|_1=(i+1)$ downwards.
    This proves (ii) and (iii).  (iv) immediately follows from the fact that there are
    $i$ copies of each vertex~$x$ with $|x|_1=i$ on level~$i$ in $T_d$.

    The number of vertices visited in phase~$i$ with $i<d$ is twice the number of edges up to depth~$i$.
    Therefore the length of phase $i<d$ is
    $2\,\sum_{j=1}^{i} j \binom{d}{j}$
    and the total length of the deterministic walk on~$H_d$
    until~$1^d$ is discovered is
    \begin{align*}
         &d + 1 + 2\,\sum_{i=0}^{d-1} \sum_{j=1}^{i} j \,\tbinom{d}{j}
        = d + 1 + 2\, \sum_{j=0}^d j\,(d-j)\,\tbinom{d}{j}
        = d + 1 + 2 d \, \sum_{j=0}^{d-1} j \,\tbinom{d-1}{j}\\
        &\qquad\qquad= d + 1 + d \, (d-1) \, 2^{d-1}
        = (n \log^2 n)/2 + \O( n \log n).
        \qedhere
    \end{align*}
\end{proof}

\section{Short Term Behavior}
\label{sec:shortterm}

\noindent
For random walks, \citet{BF93,BF96journal} examined how fast a random walk
covers a certain number of vertices and/or edges.
\tabref{shortterm} provides an overview of their bounds compared to ours.
For the deterministic walk, we can show the following result about the rate at which
the walk discovers new edges in the short term.

\newcommand{\textshortterm}{
    All deterministic walks with $\tkappa=1$
    visit~$\cN$ distinct vertices within
    $
     \min \left\{ \mbox{$\Oh(\cN \Delta + (\cN \Delta/\delta)^2 )$}  ,
                  \mbox{$\Oh( m + (m/\delta)^2 )$} \right\}
    $
    steps
    and $\cM$ distinct edges within $\Oh( \cM + ( \cM /\delta)^2 )$ steps.
}
\mythm{shortterm}{\textshortterm}

The proof of \thmref{shortterm} is based on some combinatorial property
of the deterministic walk.  Let us first observe a simple graph-theoretic lemma.

\begin{lem}\label{lem:graphfirst}\label{lem:graphsecond}
    For any graph $G=(V,E)$, vertex $v \in V$, and $i \geq 0$ with
    $\Gamma^{i+1}(v) \neq \emptyset$, we have
    \begin{align*}
        |E(\Gamma^{i}(v)) \cup E(\Gamma^{i+1}(v)) \cup E(\Gamma^{i+2}(v))  |
        &\geq \delta^2/6.
    \end{align*}
\end{lem}
\begin{proof}
    Fix a vertex $u \in \Gamma^{i+1}(v)$.
    Clearly, there exists a $j \in \{i,i+1,i+2\}$ such that
    $|E(u,\Gamma^{j}(v))| \geq \delta/3$.
    Since this implies that $|\Gamma^{j}(v)| \geq \delta/3$, we have
    $
        |E(|\Gamma^{j}(v))|
        \geq |\Gamma^{j}(v)| \delta/2
        \geq \delta^2 / 6
    $
    and the claim follows.
\end{proof}

\newcommand{\twolinebrace}{$\left\lbrace \begin{minipage}[c]{0.001mm}\vspace*{1.9\baselineskip}\end{minipage} \right.$}
\newcommand{\threelinebrace}{\scalebox{1}[1.55]{$\left\lbrace \begin{minipage}[c]{0.001mm}\vspace*{1.9\baselineskip}\end{minipage} \right.$}}

\begin{table*}[bt]
    \footnotesize
    \begin{center}
    \begin{tabular}{ll@{}l@{\ }l|l@{}l@{\ }l|}
    \cline{2-7}
    & \multicolumn{3}{|c|}{\bf Random Walk}
    & \multicolumn{3}{|c|}{\bf Deterministic Walk} \\
    \hline
    \hline
    \multicolumn{1}{|l|}{time to visit $\cN$ vertices}
    & \multirow{2}{*}{\twolinebrace}
    & $\O(\cN^3)$
    & \cite[Thm.~1.1]{BF96journal}
    & \multirow{2}{*}{\twolinebrace}
    & $\Oh(\cN \Delta + (\cN \Delta/\delta)^2 )$
    & (\shortthmref{shortterm})
    \\
    \multicolumn{1}{|l|}{on arbitrary graphs}
    &
    & $\O(m \cN)$
    & \cite[Thm.~1.4]{BF96journal}
    &
    & $\Oh( m + (m/\delta)^2 )$
    & (\shortthmref{shortterm})
    \\

    \hline
    \multicolumn{1}{|l|}{\multirow{3}{*}{\begin{tabular}{@{}l}time to visit $\cM$ edges\\
        on arbitrary graphs\end{tabular}}}
    & \multirow{3}{*}{\threelinebrace}
    & $\O(\cM^2)$
    & \cite[Thm.~1.2]{BF96journal}
    &
    & \multirow{3}{*}{$\Oh( \cM + ( \cM /\delta)^2 )$}
    & \multirow{3}{*}{(\shortthmref{shortterm})}
    \\
    \multicolumn{1}{|l|}{}
    && $\O(n \cM)$
    & \cite[Thm.~1.4]{BF96journal}
    & \multicolumn{3}{|c|}{}\\
    \multicolumn{1}{|c|}{}
    && $\O( \cM + ( \cM^2 \log \cM)/\delta )$
    & \cite[Thm.~5]{BF93}
    & \multicolumn{3}{|c|}{}
    \\
    \hline
    \end{tabular}
    \end{center}
    \caption{Short term behavior of random and deterministic walk. For the time
    to cover $\cN$ vertices, the bounds for the random walk are always as good as the corresponding bounds for the deterministic walk.
    The two respective upper bounds for the time to cover~$\cM$ edges are incomparable
    in general.
    }
    \vspace*{-1.5\baselineskip}
    \label{tab:shortterm}
\end{table*}

In the proof of \thmref{shortterm} we will also need the following
property borrowed from \citet{YanovskiWB03}.
\begin{lem}
    \label{lem:balance}\label{lem:balanceedge}
    For any time $t$ and edges $\{u,v\},\{v,w\} \in E$ it holds that
    \begin{align*}
        |\tN_t(u \to v) - \tN_t(v \to w)| &\leq 2.
    \end{align*}
\end{lem}

\begin{proof}[Proof of \thmref{shortterm}]
We start with the second claim. Assume that there is an edge $e=(u,v)\in E$ with $\tN_t(e) \geq 13 \sqrt{t}/\delta$. Then we know that for all adjacent edges $e'=(v,w)$ that $\tN_t(e') \geq 13 \sqrt{t}/\delta-2$ as $|\tN_t(e) - \tN_t(e')| \leq 2$
by \cite[Cor.~4]{YanovskiWB03}.
More generally, for an edge $e=(x,y)$ with $\dist(u,x)=i$ we have
\begin{equation}
    \label{eq:alledges}
    \tN_t(x\to y) \geq 13 \sqrt{t}/\delta - 2 i.
\end{equation}
As in the proof
of the first claim, we may assume that
$\Gamma^i(u) \neq \emptyset$ for $1 \leq i \leq 2 \sqrt{t}/\delta$,
as otherwise all edges have been traversed already by \eq{alledges}.
With \lemref{graphsecond},
\begin{align*}
  t &= \sum_{\{u,v\} \in E} \tN_{t}(u\to v) \\
    &\geq \sum_{k=0}^{2 \sqrt{t}/(3\delta)}
        |E(\Gamma^{3k}(v)) \cup E(\Gamma^{3k+1}(v)) \cup E(\Gamma^{3k+2}(v)) |
        \big( 13 \sqrt{t} / \delta - 6 (k+1) \big) \\
    &> \sum_{k=0}^{2 \sqrt{t}/(3\delta) - 1}  \frac{\delta^2}{6}
        \left( 13 \sqrt{t} /\delta - 6 (k+1)  \right)
    = \frac{1}{6} \delta^2 \bigg( \frac{2\sqrt{t}}{3 \delta} \bigg) 9 \frac{ \sqrt{t} }{\delta}
    = t,
\end{align*}
gives a contradiction.
Therefore we conclude that no edge is visited more often than $\max\{1, 13
\sqrt{t}/\delta \}$. Hence after $t$ steps, at least $\min \{t, \sqrt{t} \delta/13 \}$
distinct edges must be visited. Choosing $t:= \max \big\{ \cM, 13^2 \, \cM / \delta^2  \big\}$,
the deterministic walk visits at least $\cM$ distinct edges.

For the first claim, we observe that if the random walk visits at least $\cM$ distinct edges, then it also visits at least $\cM/\Delta$ distinct vertices. Hence to visit $\cN$ vertices, we have to visit at least $\min \{\cN \Delta, m \}$ edges. As shown above, to visit $\min \{\cN \Delta, m \}$ edges, we have to spend \[ \min \left\{ \Oh(\cN \Delta + (\cN \Delta/\delta)^2 )  , \Oh( m + (m/\delta)^2 ) \right\}
\] steps.
\end{proof}

At the end of \secref{ourresults}, we give an example how the bounds for the
random and deterministic walk compare to each other.
For expander graphs, much stronger results
are known.  There, $\max_{u,v} \rH(u,v) = \Oh(n)$ and hence by a simple first-moment calculation one obtains that after $\Oh(n)$ steps, $cn$ vertices are visited (where~$0 < c < 1$ is any constant).  We remark that a similar result can be shown for
the deterministic walk, that is, after $\Oh(n \log \log n)$ steps
it visits $cn$ vertices of an expander graph (of constant degree).

\section{Discussion}

\noindent
We have analyzed the vertex and edge cover time of the deterministic random walk
and presented upper bounds for general graphs
based on the local divergence, expansion properties, and flows.
This is complemented with tight bounds for various common graph classes.
It turns out that the deterministic random walk is surprisingly efficient though
it has a strong adversary (as the order of the rotors is completely arbitrary) and
it is not tailored for search problems (as it does not mark visited vertices).

In applications such as
broadcasting~\citep{DFS08,DFS09,DFKS09,ADHP2009,HF2009}
and sorting~\citep{BarveGV97}
the quasirandom version of the
deterministic random walk seems to be especially efficient.
There, the first rotor direction is chosen at random while the order of the
rotors stays arbitrary.
It would be interesting to quantify how much this quasirandom walk
can cover all vertices or edges faster than our deterministic one.

\newcommand{\FOCS}[2]{#1 Annual IEEE Symposium on Foundations of Computer Science (FOCS '#2)}
\newcommand{\STOC}[2]{#1 Annual ACM Symposium on Theory of Computing (STOC '#2)}
\newcommand{\SODA}[2]{#1 Annual ACM-SIAM Symposium on Discrete Algorithms (SODA '#2)}
\newcommand{\ICALP}[2]{#1 International Colloquium on Automata, Languages, and Programming (ICALP '#2)}
\newcommand{\ISAAC}[2]{#1 International Symposium on Algorithms and Computation (ISAAC '#2)}

\end{document}

%% file: lowertorus_initial.tex
\begin{tikzpicture}[scale=.39,baseline=0pt]
\filldraw[color=white] (-3.1,-3.1) rectangle (3.1,3.1);
\draw[step=1cm,thin,color=black!30] (-2.7,-2.7) grid (2.7,2.7);
\draw[->,thick] (-2,-2) -- (-1.6,-2);
\draw[->,thick] (-2,-1) -- (-1.6,-1);
\draw[->,thick] (-2,0) -- (-1.6,0);
\draw[->,thick] (-2,1) -- (-1.6,1);
\draw[->,thick] (-2,2) -- (-2,1.6);
\%draw[->,thick] (-2,3) -- (-2,2.6);
\draw[->,thick] (-1,-2) -- (-1,-1.6);
\draw[->,thick] (-1,-1) -- (-0.6,-1);
\draw[->,thick] (-1,0) -- (-0.6,0);
\draw[->,thick] (-1,1) -- (-1,0.6);
\draw[->,thick] (-1,2) -- (-1,1.6);
\%draw[->,thick] (-1,3) -- (-1,2.6);
\draw[->,thick] (0,-2) -- (0,-1.6);
\draw[->,thick] (0,-1) -- (0,-0.6);
\draw[->,thick] (0,0) -- (0,0.4);
\draw[->,thick] (0,1) -- (0,0.6);
\draw[->,thick] (0,2) -- (0,1.6);
\%draw[->,thick] (0,3) -- (0,2.6);
\draw[->,thick] (1,-2) -- (1,-1.6);
\draw[->,thick] (1,-1) -- (1,-0.6);
\draw[->,thick] (1,0) -- (0.6,0);
\draw[->,thick] (1,1) -- (0.6,1);
\draw[->,thick] (1,2) -- (1,1.6);
\%draw[->,thick] (2,-3) -- (2,-2.6);
\draw[->,thick] (2,-2) -- (2,-1.6);
\draw[->,thick] (2,-1) -- (1.6,-1);
\draw[->,thick] (2,0) -- (1.6,0);
\draw[->,thick] (2,1) -- (1.6,1);
\draw[->,thick] (2,2) -- (1.6,2);
\foreach \x in {-2,...,2} \foreach \y in {-2,...,2} {
\filldraw[color=black] (\x,\y) circle (.07cm);
}
\filldraw[color=blue] (0,0) circle (.11cm);
\end{tikzpicture}

%% file: lowertorus_in.tex
\begin{tikzpicture}[scale=.39,baseline=0pt]
\foreach \x in {-2,...,2} \foreach \y in {-2,...,2} {
\filldraw[color=black!30] (\x,\y) circle (.07cm);
}
\draw[step=1cm,thin,color=black!30] (-2.7,-2.7) grid (2.7,2.7);
\filldraw[color=black] (-2,-2) circle (.07cm);
\filldraw[color=black] (-2,-1) circle (.07cm);
\filldraw[color=black] (-2,0) circle (.07cm);
\filldraw[color=black] (-2,1) circle (.07cm);
\filldraw[color=black] (-2,2) circle (.07cm);
\filldraw[color=black] (-1,-2) circle (.07cm);
\filldraw[color=black] (-1,2) circle (.07cm);
\filldraw[color=black] (0,-2) circle (.07cm);
\filldraw[color=black] (0,2) circle (.07cm);
\filldraw[color=black] (1,-2) circle (.07cm);
\filldraw[color=black] (1,2) circle (.07cm);
\filldraw[color=black] (2,-2) circle (.07cm);
\filldraw[color=black] (2,-1) circle (.07cm);
\filldraw[color=black] (2,0) circle (.07cm);
\filldraw[color=black] (2,1) circle (.07cm);
\filldraw[color=black] (2,2) circle (.07cm);
\draw[->,thick] (-2,-2) -- (-1.4,-2);
\draw[->,thick] (-2,-1) -- (-1.4,-1);
\draw[->,thick] (-2,0) -- (-1.4,0);
\draw[->,thick] (-2,1) -- (-1.4,1);
\draw[->,thick] (-2,2) -- (-2,1.4);
\draw[->,thick] (-1,-2) -- (-1,-1.4);
\draw[->,thick] (-1,2) -- (-1,1.4);
\draw[->,thick] (0,-2) -- (0,-1.4);
\draw[->,thick] (0,2) -- (0,1.4);
\draw[->,thick] (1,-2) -- (1,-1.4);
\draw[->,thick] (1,2) -- (1,1.4);
\draw[->,thick] (2,-2) -- (2,-1.4);
\draw[->,thick] (2,-1) -- (1.4,-1);
\draw[->,thick] (2,0) -- (1.4,0);
\draw[->,thick] (2,1) -- (1.4,1);
\draw[->,thick] (2,2) -- (1.4,2);
\filldraw[color=blue] (0,0) circle (.11cm);
\end{tikzpicture}

%% file: lowertorus_cycle.tex
\begin{tikzpicture}[scale=.39,baseline=0pt]
\foreach \x in {-2,...,2} \foreach \y in {-2,...,2} {
\filldraw[color=black!30] (\x,\y) circle (.07cm);
}
\draw[step=1cm,thin,color=black!30] (-2.7,-2.7) grid (2.7,2.7);
\filldraw[color=black] (-2,-2) circle (.07cm);
\filldraw[color=black] (-2,-1) circle (.07cm);
\filldraw[color=black] (-2,0) circle (.07cm);
\filldraw[color=black] (-2,1) circle (.07cm);
\filldraw[color=black] (-2,2) circle (.07cm);
\filldraw[color=black] (-1,-2) circle (.07cm);
\filldraw[color=black] (-1,2) circle (.07cm);
\filldraw[color=black] (0,-2) circle (.07cm);
\filldraw[color=black] (0,2) circle (.07cm);
\filldraw[color=black] (1,-2) circle (.07cm);
\filldraw[color=black] (1,2) circle (.07cm);
\filldraw[color=black] (2,-2) circle (.07cm);
\filldraw[color=black] (2,-1) circle (.07cm);
\filldraw[color=black] (2,0) circle (.07cm);
\filldraw[color=black] (2,1) circle (.07cm);
\filldraw[color=black] (2,2) circle (.07cm);
\draw[->,thick] (-2,-2) -- (-1.4,-2);
\draw[->,thick] (-2,-1) -- (-2,-1.6);
\draw[->,thick] (-2,0) -- (-2,-0.6);
\draw[->,thick] (-2,1) -- (-2,0.4);
\draw[->,thick] (-2,2) -- (-2,1.4);
\draw[->,thick] (-1,-2) -- (-0.4,-2);
\draw[->,thick] (-1,2) -- (-1.6,2);
\draw[->,thick] (0,-2) -- (0.6,-2);
\draw[->,thick] (0,2) -- (-0.6,2);
\draw[->,thick] (1,-2) -- (1.6,-2);
\draw[->,thick] (1,2) -- (0.4,2);
\draw[->,thick] (2,-2) -- (2,-1.4);
\draw[->,thick] (2,-1) -- (2,-0.4);
\draw[->,thick] (2,0) -- (2,0.6);
\draw[->,thick] (2,1) -- (2,1.6);
\draw[->,thick] (2,2) -- (1.4,2);
\filldraw[color=blue] (0,0) circle (.11cm);
\end{tikzpicture}

%% file: lowertorus_out.tex
\begin{tikzpicture}[scale=.39,baseline=0pt]
\foreach \x in {-2,...,2} \foreach \y in {-2,...,2} {
\filldraw[color=black!30] (\x,\y) circle (.07cm);
}
\draw[step=1cm,thin,color=black!30] (-2.7,-2.7) grid (2.7,2.7);
\filldraw[color=black] (-2,-2) circle (.07cm);
\filldraw[color=black] (-2,-1) circle (.07cm);
\filldraw[color=black] (-2,0) circle (.07cm);
\filldraw[color=black] (-2,1) circle (.07cm);
\filldraw[color=black] (-2,2) circle (.07cm);
\filldraw[color=black] (-1,-2) circle (.07cm);
\filldraw[color=black] (-1,2) circle (.07cm);
\filldraw[color=black] (0,-2) circle (.07cm);
\filldraw[color=black] (0,2) circle (.07cm);
\filldraw[color=black] (1,-2) circle (.07cm);
\filldraw[color=black] (1,2) circle (.07cm);
\filldraw[color=black] (2,-2) circle (.07cm);
\filldraw[color=black] (2,-1) circle (.07cm);
\filldraw[color=black] (2,0) circle (.07cm);
\filldraw[color=black] (2,1) circle (.07cm);
\filldraw[color=black] (2,2) circle (.07cm);
\draw[->,thick] (-2,-2) -- (-2,-2.6);
\draw[->,thick] (-2,-1) -- (-2.6,-1);
\draw[->,thick] (-2,0) -- (-2.6,0);
\draw[->,thick] (-2,1) -- (-2.6,1);
\draw[->,thick] (-2,2) -- (-2.6,2);
\draw[->,thick] (-1,-2) -- (-1,-2.6);
\draw[->,thick] (-1,2) -- (-1,2.6);
\draw[->,thick] (0,-2) -- (0,-2.6);
\draw[->,thick] (0,2) -- (0.6,2);
\draw[->,thick] (1,-2) -- (1,-2.6);
\draw[->,thick] (1,2) -- (1,2.6);
\draw[->,thick] (2,-2) -- (2.6,-2);
\draw[->,thick] (2,-1) -- (2.6,-1);
\draw[->,thick] (2,0) -- (2.6,0);
\draw[->,thick] (2,1) -- (2.6,1);
\draw[->,thick] (2,2) -- (2,2.6);
\filldraw[color=blue] (0,0) circle (.11cm);
\end{tikzpicture}